\RequirePackage{fix-cm}
\documentclass[smallextended]{svjour3}       
\smartqed  
\usepackage{graphicx}
\usepackage{graphics}
\usepackage{amsmath,xcolor}
\usepackage[round]{natbib}
\usepackage{lineno}

\newcommand{\RNA}{{\mbox{RNA}}}
\newcommand{\Tat}{{\mbox{Tat}}}
\newcommand{\LTR}{{\mbox{LTR}}}

\newcommand{\TAR}{{\mbox{TAR}}}
\newcommand{\p}{{\mbox{p}}}
\newcommand{\env}{{\mbox{\textit{env}}}}

\newcommand{\beanum}{\begin{eqnarray}}
\newcommand{\eeanum}{\end{eqnarray}}
\newcommand{\bea}{\begin{eqnarray}}
\newcommand{\eea}{\end{eqnarray}}
\newcommand{\bsa}{\begin{subeqnarray}}
\newcommand{\esa}{\end{subeqnarray}}
\newcommand{\be}{\begin{equation}}
\newcommand{\ee}{\end{equation}}

\begin{document}

\title{Mathematical analysis and potential therapeutic implications of a novel HIV-1 model of basal and activated transcription in T-cells and macrophages}



\author{Tin Phan$^1$                 \and
        Catherine DeMarino$^2$       \and 
        Fatah Kashanchi$^2$          \and
        Yang Kuang$^1$               \and
        Daniel M. Anderson$^3$       \and
        Maria Emelianenko$^3$        
}

\authorrunning{T. Phan \textit{et al}.} 

\institute{
        Tin Phan \at
              \email{tin.t.phan@asu.edu}           
          \and
        Maria Emelianenko \at
              \email{memelian@gmu.edu}           
            \and
        $^1$ School of Mathematical and Statistical Sciences, Arizona State University, Tempe, AZ, USA \\
        $^2$ Laboratory of Molecular Virology, School of Systems Biology, George Mason University, Manassas, VA, USA \\
        $^3$ Department of Mathematical Sciences, George Mason University, Fairfax, VA, USA
}

\date{Received: date / Accepted: date}

\maketitle

\begin{abstract}

HIV-1 affects tens of millions of people worldwide.
Current treatments often involve a cocktail of antiretroviral drugs, which are effective in reducing the virus and extending life spans.
However, there is currently no FDA-approved HIV-1 transcription inhibitor. Furthermore, there have only been a few attempts to model the transcription process in HIV-1.
In this work, we extend a novel three-state model of HIV-1 transcription introduced in \citep{demarino2020differences} that has been developed and validated against experimental data. After fitting this model to \textit{in vitro} data, significant differences in the transcription process of HIV-1 in T-cells and macrophages have been observed. In particular, the activation of the HIV-1 promoter in T-cells appears to take place rapidly as the Tat protein approaches a critical threshold. In contrast, the same process occurs smoother in macrophages.

In this work, we carry out systematic mathematical analyses of the model to complement experimental data fitting and sensitivity analysis performed earlier.
We derive explicit solutions of the model to obtain exact transcription process decay rates for the original model and then study the effect of nonlinearity on the system behavior, including the existence and the local and global stability of the positive equilibrium. We were able to show the stability of the positive steady state in limiting cases, with the global stability in the general case remaining an open question. 

By modeling the effect of transcription-inhibiting drug therapy, we provide a nontrivial condition for it to be effective in reducing viral load. Moreover, our numerical simulations and analysis point out that the effect of the transcription-inhibitor can be enhanced by synchronizing with standard treatments, such as combination antiretroviral therapy, to allow the reduction of total dosages and toxicity.
%

\keywords{HIV-1 transcription \and F07\#13 \and Combination antiretroviral therapy \and Transcription-inhibitor \and Treatment combination}
\end{abstract} 

\section{Introduction}
\label{intro}
The Human Immunodeficiency Virus Type I (HIV-1) is the causative agent of acquired immune deficiency syndrome (AIDS). Since the advent of combination antiretroviral therapy (cART) in the early 1990s, infected individuals are living longer, healthier lives and transmission rates have slowed down.
However, of the 36.9 million infected worldwide, only 21.7 million people were reported to be accessing cART and even fewer were maintaining the strict adherence required by the therapy~\citep{unaidsdate_2019}. 
Current cART regimens have been developed to target HIV-1       at almost every stage of the viral life cycle. These include fusion/entry inhibitors which target HIV-1 cellular receptors and associated viral proteins, reverse transcriptase inhibitors which prevent the production of viral DNA from RNA, integrase inhibitors that function to mitigate integration of the virus into host DNA, and protease inhibitors which block maturation of viral proteins. 
The combination of several inhibitors is effective in lowering viral titers and reducing morbidity and mortality in infected individuals~\citep{heaton2010hiv,deeks2013end,mothobi2012neurocognitive}.
However, to date, there are no FDA-approved antiretrovirals that target HIV-1 transcription. 
This therapeutic gap leads to persistent, low-level viral transcription despite suppressive treatment, a concept which has been termed ``leaky latency'', resulting in approximately $1\times10^3$ copies of cell-associated viral RNA in infected cells~\citep{furtado1999persistence,hatano2012cell,kumar2007human}. While viremia is adequately controlled ($<50$ viral RNA/mL), cell-associated viral RNA can contribute to chronic inflammation, rapid viral rebound, immune dysfunction via direct mechanisms, stochastic production of viral proteins, or via release of viral RNA in extracellular vesicles~\citep{mccauley2018intron,akiyama2018hiv,ferdin2018viral,narayanan2013exosomes,sampey2016exosomes,demarino2018antiretroviral,hladnik2017trans,li2016size,henderson2019presence}.

Despite the presence of cART, HIV-1       can persist in viral reservoirs including long-lived memory CD4$^+$ T-cells, blood-brain barrier protected myeloid cells of the central nervous system (CNS), and low cART penetration lymphoid tissues such as lymph nodes and gut-associated lymphoid tissue (GALT)~\citep{sengupta2018targeting,li2016astrocytes,garrido2015translational,dave2018follicular,hatano2013comparison}. 
These reservoirs can be maintained through several mechanisms including chromatin modifications and blocks in viral transcription initiation and elongation. Although HIV-1 can persist in a latent state for long periods, activation of latently infected cells through antigen stimulation or cytokine activation can lead to the induction of HIV-1 transcription factors such as nuclear factor kappa-light-chain-enhancer of activated B cells (NF-$\kappa$B) or nuclear factor of activated T-cells (NFAT).
Production of these transcription factors can, in turn, cause viral reactivation of the latent provirus leading to transcription of the HIV-1 genome and subsequent production of viral proteins \citep{chou2013hiv,mbonye2014transcriptional,kumar2014hiv}.
Importantly, activation of the virus elicits cytolysis and immune-mediated responses to clear the virus. This mechanism has led to the development of a therapeutic strategy termed ``shock and kill'', an approach which takes aim at latency-mediating mechanisms, such as histone deacetylases (HDACs)~\citep{archin2009expression,lehrman2005depletion,wei2014histone}, using latency-reversing agents (LRAs) to reactivate and promote immune clearance of the virus.
Conversely, others have proposed an opposite strategy known as ``lock and block'' which focuses on promoting an inactive state of the HIV-1       LTR to inhibit viral transcription and virion production through the use of latency promoting agents (LPAs). 
These studies have led to the identification of several HIV-1       transcription inhibitors which have shown success \textit{in vitro, in vivo}, and in clinical trials~\citep{mousseau2012analog,mousseau2015tat,kim2016inhibition,jean2017curaxin,kessing2017vivo,rutsaert2019safety,hayashi2017screening}.

While there is a rich literature of mathematical modeling for HIV-1 transmission at the population level 
\citep{eaton2012hiv,omondi2018mathematical,velasco2002could,li2018mathematical} 
and its interaction with the immune system with or without treatments
\citep{perelson1999mathematical,wodarz2002mathematical,wang2016dynamics,adak2018analysis},
mathematical models at the molecular level for HIV-1 are far and few between. 
More recently, 
\cite{chavali2015distinct}
developed and showed that a multi-state promoter HIV-1 model is perhaps better at capturing the heterogeneous reactivation of HIV-1 in response to treatments (e.g. ``shock and kill'' therapy) compared to the single-state promoter model. 
\cite{ke2015modeling}
also used a multi-state promoter model to study the effect of Vorinostat, a drug used in the activation of HIV-1 transcription, with experimental data. 
Additionally, \cite{gupta2018trade} utilized a multi-state promoter model to study the synergy in combination of latency-reversing therapies using stochastic simulation. 

In prior work, we have developed a three-state LTR model of HIV-1   transcription~\citep{demarino2020differences}. 
    This model evaluates various states of the HIV-1       LTR, repressed ($\LTR_R$), intermediate ($\LTR_I$), and activated ($\LTR_A$) in response to various stimuli including transcription inducers. 
    Furthermore, we modeled the transcription of two viral RNAs; a short non-coding RNA trans-activation response (TAR) element, and genomic RNA (\textit{env}), as well as levels of viral proteins Tat (trans-activator of transcription), Pr55, and p24 in response to changes in the LTR state. 
    This model has been validated in two types of immune cells, T-cells, and myeloids, using biochemical assays which assess each parameter and model predictions at extended time frames. 
    In contrast to the model by \cite{chavali2015distinct}, \cite{demarino2020differences} took into account that the production of $\Tat$, an early HIV-1 protein, only occurs during the intermediate activation state, while genomic RNA produced during the activated state is used to facilitate the production of viral particles.
    The simple structure of the model allows for direct incorporation of various therapies, which potentially serves as a valuable tool in evaluating viral transcription in response to various stimuli including LRAs and LPAs. 

In this work, we carry out systematic mathematical analyses of the previous model to show its biological and mathematical validity.
The original model formulation utilizes a switching function to model the $\Tat$-dependent functional responses, which limits the ability of the model to characterize differences in the transcriptional behaviors in T-cells and macrophages. Additionally, it leads to a discontinuity, which may not be biologically relevant.
To address this issue, we extend the model to consider continuous $\Tat$-dependent functional responses.
Similarly, we carry out mathematical analyses and data fitting for the new model.
We provide stability results in several limiting cases; however, the global stability of the positive steady state in the general case is still an open question.
By comparing the two models (switching vs. continuous response), we find observations that provide insights into the transcription process of HIV-1.
Specifically, there is a clear distinction in the transcriptional behaviors between T-cells and macrophages, which is dependent on the amount of $\Tat$ protein.
Finally, we use the model to study the effectiveness of an experimental transcription-inhibitor drug.
Our results suggest that the Tat peptide mimetic transcriptional inhibitor (F07\#13) is synchronous with standard treatment. Thus, by combining F07\#13 with standard treatment such as cART, the total dosage and the potential side effects may be reduced.

The remainder of this paper is organized as follows. 
In section~\ref{sec:math_model}, we briefly introduce the mathematical model and describe how some treatments of HIV-1 are incorporated into the model.
In section~\ref{sec:analysis}, we carry out the study of the properties of the model and its extension. The main goal is to demonstrate that the model exhibits the expected biological dynamics, which entails the analyses of positive invariance, boundedness and stability of solutions.
Utilizing experimental data in \cite{demarino2020differences}, we carry out parameter estimation to differentiate the dynamics between T-cells and macrophages in section~\ref{sec:param_esti}.
An important novelty of the model is its ability to incorporate multiple treatments of HIV-1. 
Therefore, in section~\ref{sec:drug_effect}, we study the effect of the transcriptional inhibitor F07\#13 in combination with other drugs.
Finally, we discuss our results in section~\ref{sec:disscussion}. Derivation of closed form solution of the linear model and details of numerical parameter estimation are provided in the Appendix.

\section{Mathematical model}
\label{sec:math_model}
In \cite{demarino2020differences}, the following three-state HIV-1 model of the transcription process was derived and validated against experimental data in T-cells and macrophages.
\begin{figure*}[!ht]
\centering
\includegraphics[scale=.6]{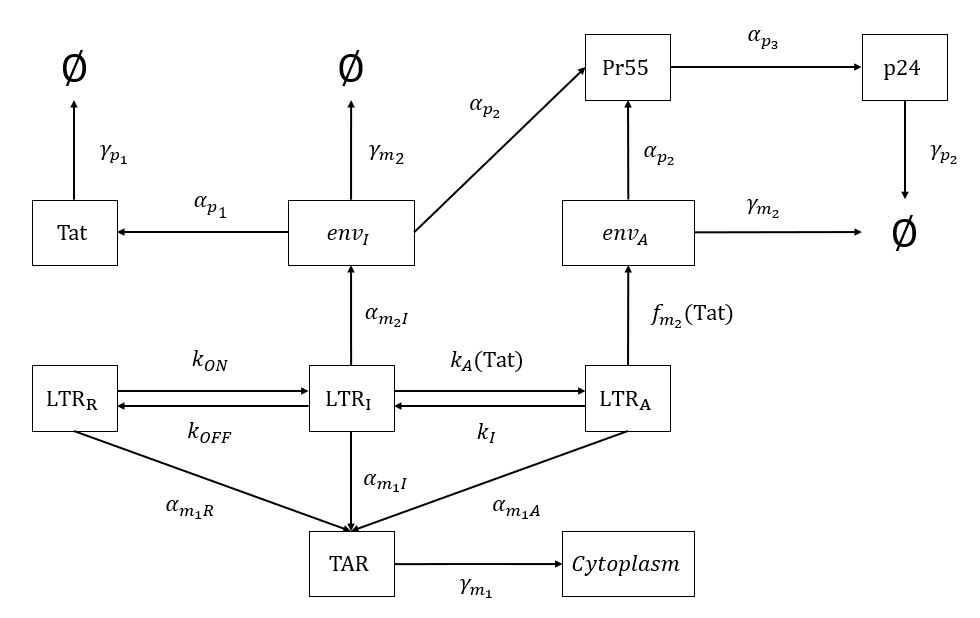}
\caption{HIV-1 transcription model. Three States of $\LTR$, e.g. repressed ($\LTR_R$), intermediate ($\LTR_I$), activated ($\LTR_A$). The long and short RNA ($\env_I$, $\env_A$ and $\TAR$), $\Tat$, $\Pr55$ and p$24$ interactions. The figure is adapted from \cite{demarino2020differences} under a Creative Commons Attribution (CC BY) license.}
\label{fig:diagram}
\end{figure*}
%
\begin{eqnarray}
\label{eq:main}
 \frac{d}{dt} \LTR_R &  =  & k_{OFF} \LTR_I - k_{ON} \LTR_R \\
 \label{eq: LTR_I}
 \frac{d}{dt} \LTR_I   & = &  - \left[\frac{w_3 w_4}{w_5} k_A(\Tat) + k_{OFF} \right] \LTR_I + w_1 k_I \LTR_A  + k_{ON} \LTR_R  \\
 \label{eq: LTR_A}
 \frac{d}{dt} \LTR_A   & = & \frac{w_3 w_4}{w_5} k_A(\Tat) \;\LTR_I - w_1 k_I  \LTR_A \label{eq:main3} 
 \\
 \label{eq: Tat}
 \frac{d}{dt} \Tat & = &  \alpha_{p_1} \env_{I} - \gamma_{p_1}  \Tat  \\
 \frac{d}{dt} \TAR   & = &    \alpha_{m_1,R} \LTR_R  + \alpha_{m_1,I} \LTR_I + \alpha_{m_1,A} \LTR_A - \gamma_{m_1} \TAR  \\
 \label{eq: env_I}
 \frac{d}{dt} \env_{I}  & = &   \alpha_{m_2,I} \LTR_I - \left( \gamma_{m,2} + \alpha_{p_1} + \alpha_{p_2} \right) \env_{I}  \\ 
 \frac{d}{dt} \env_{A} & = &   f_{m_2} (\Tat) \LTR_A - \left( \gamma_{m,2} + \alpha_{p_2} \right) \env_{A}   \\
 \frac{d}{dt}\Pr55 & = & \alpha_{p_2} \env_{I} + \alpha_{p_2} \env_{A} - (\alpha_{p_3}/w_2) P_{r55}  \\
 \label{eq:main_end}
 \frac{d}{dt}\p24 & = & (\alpha_{p_3}/w_2) P_{r55}  - \gamma_{p_2} p_{24} 
\end{eqnarray}

The model incorporates important features of the basal and activated transcription of the HIV-1        genome. 
The Long Terminal Repeat (LTR) is categorized into three stages, suppressed ($\LTR_R$), intermediate ($\LTR_I$) and activated ($\LTR_A$), similar to that of the model in \cite{chavali2015distinct}.
The `OFF' states refer to the repressed and intermediate HIV-1        promoters, while the `ON' state refers to the activated HIV-1        promoter. Here, $k_{ON}, k_{OFF}, k_A(\Tat)$ and $k_I$ are transition rates from one state of LTR to another as indicated in Figure~\ref{fig:diagram}. The total LTR is assumed to be conserved, so each LTR state in the model represents the proportion of LTR in the respective state.

Another key feature of the model is the division of viral RNA into short-non-coding RNA and long-genomic RNA, which are characterized by the amount of TAR and \textit{env}, respectively. The model further divides the \textit{env} according to its promoter LTR state, or $\env_I$ and $\env_A$ corresponding to $\LTR_I$ and $\LTR_A$, respectively.
Additionally, TAR is produced by all three states of LTR, but generally at different rates.
$\Tat$ is produced via the translation of a multiply-spliced mRNA, which is represented in the model as the transition from $\env_I$ to $\Tat$. The presence of $\Tat$ directly affects the activation rate of the intermediate LTR state. Hence, The coefficient $k_A(\Tat)$ is expected to depend on the level of $\Tat$. 
Due to the quick transition to the active state, this can be approximated as a step function or a Hill function (see Section \ref{sec:analysis}). The value of $\Tat_{crit}$ is the estimated number of $\Tat$ required to overcome potential sequestration by $\TAR$ in the cytoplasm to allow for efficient Tat-activated transcription.
$\Tat$ further enhances the transcription rate of activated LTR to produce $\env_A$. This rate can also be represented by a Hill function with the same $\Tat_{crit}$ value as $k_A(\Tat)$.
Finally, both $\env_I$ and $\env_A$ are used to produce the HIV-1 gag polyprotein, $\Pr55$. Following its production, $\Pr55$ is cleaved into smaller proteins, one of which is $\p24$, which forms the capsid and is tractable experimentally.
Additional details on the parameter values are listed in Table~\ref{table:parameters}.  

\begin{table}
\begin{center}
\begin{tabular}{lcccc}
\hline
                 & Unit                              & definition                  & T-Cell    & macrophages \\ \hline\hline
$k_{ON}$         & [H1 pc/mL/hr] & $\LTR_R \rightarrow \LTR_I$ & $5.785\%$ & $9.245\%$ \\[.75 ex]
$k_{OFF}$        & [H1 pc/mL/hr] & $\LTR_I \rightarrow \LTR_R$ & $1.220\%$ & $1.228\%$ \\[.75 ex]
$k_{A}(\Tat)$    & [H1 pc/mL/hr] & $\LTR_I \rightarrow \LTR_A$ & $3.409\%$ & $9.010\%$ \\[.75 ex]
$k_{I}$          & [H1 pc/mL/hr] & $\LTR_A \rightarrow \LTR_I$ & $0.0\%$   & $2.451\%$ \\[.75 ex]
$\alpha_{m_1,R}$ & [copies/mL/hr]   & $\LTR_R \rightarrow \TAR$   & $2.50 \times 10^4$   & $2.90 \times 10^4$ \\[.75 ex]
$\alpha_{m_1,I}$ & [copies/mL/hr]   & $\LTR_I \rightarrow \TAR$   & $2.80 \times 10^8$   & $7.54 \times 10^4$ \\[.75 ex]
$\alpha_{m_1,A}$ & [copies/mL/hr]   & $\LTR_A \rightarrow \TAR$   & $1.37 \times 10^7$   & $4.51 \times 10^5$ \\[.75 ex]
$\alpha_{m_2,I}$ & [copies/mL/hr]   & $\LTR_I \rightarrow \env_I$ & $3.63 \times 10^5$   & $8.13 \times 10^3$ \\[.75 ex]
$\alpha_{m_2,A}$ & [copies/mL/hr]   & $\LTR_A \rightarrow \env_I$ & $2.47 \times 10^6$   & $4.00 \times 10^4$ \\[.75 ex]
$\alpha_{p_1}$   & [$\Tat$ dc/mL/hr]  & $\TAR \rightarrow \Tat$  & $0.040\%$  & $0.038\%$ \\[.75 ex]
$\alpha_{p_2}$   & [$\Pr55$ dc/mL/hr] & $\TAR \rightarrow \Pr55$ & $0.154\%$  & $0.194\%$ \\[.75 ex]
$\alpha_{p_3}$   & [$p24$ dc/mL/hr]   & $\Pr55 \rightarrow \p24$  & $0.136\%$  & $0.081\%$ \\[.75 ex]
$\gamma_{p_1}$   & [Tat/mL/hr]      & $\Tat$ degradation   & $0^*$    & $0^*$ \\[.75 ex]
$\gamma_{p_2}$   & [$P_{24}$/mL/hr] & $p24$ degradation    & $0^*$    & $0^*$ \\[.75 ex]
$\gamma_{m,1}$   & [$\TAR$/mL/hr]   & $\TAR$ degradation   & $1.17 \times 10^4$    & $2.68 \times 10^4$ \\[.75 ex]
$\gamma_{m,2}$   & [$\env_I$/mL/hr] & $\env_I$ degradation & $2.24 \times 10^3$    & $5.91 \times 10^2$ \\[.75 ex]
$\Tat_{crit}$    & [Tat]            & switching limit      & $3001$ & $3001$ \\[.75 ex]
$v_a$            & unitless         & folds of increasing  & $150$  & $150$ \\ [.75 ex]
 \hline
 \end{tabular} \\
 \end{center}
\vspace{0.0in}
\caption{Table of parameters adapted from DeMarino \textit{et al}.~\citep{demarino2020differences}. 
Here, ``pc" - phosphorylation change, ``dc'' - densitometry change.
DeMarino and colleagues assumed the degradation rates for $\Tat$ and $\p24$ are taken to be $0^*$ for short period of time, but for our analytical purpose, we will assume them to be strictly positive.}
\label{table:parameters}
\end{table} 

One of the novel usages of this model is its ability to incorporate and study the effects of different drugs~\citep{demarino2020differences}. In the model, the parameter $w_i$ (i = 1,2,3,4,5) refers to the effect of different drugs on the transcriptional dynamics of HIV-1. The current form of the model shows a possibility of incorporating multiple drugs to study their effect in combination with each other ($w_2$ - cART, $w_3$ - IR, $w_4$ PMA/PHA); however, for this work, we will focus on a particular drug F07\#13.
The drug F07\#13, a Tat peptide mimetic, was developed to inhibit the transcription of HIV-1 virus by inhibiting Tat transactivation of the HIV-1 promoter, thereby encouraging the reverse direction from $\LTR_A$ to $\LTR_I$ and suppressing the activation of $\LTR_I$~\citep{lin2017inhibition,van2013effect}.
In the model, these effects are represented by the parameters $w_1$ and $w_5$, respectively.

The model is used with the following initial conditions: $\LTR_R =1$ (or $100\%$) and $\LTR_I = \LTR_A=0$, which means that we assume that all LTRs are in the repressed state initially.
Also, $\RNA_1 =0$, $\RNA_{2I}=0$, $\RNA_{2A}=0$, $\Tat =0$, $\Pr55=0$ and $\p24=0$ at $t=0$. 

\section{Dynamical system analysis}
\label{sec:analysis}
\subsection{Global stability in case of piecewise constant switching rate $f_m(\Tat)$ and constant activation rate $k_A$}
\label{ssec:analysis_linear_switching_case}

In DeMarino \textit{et al}.~\citep{demarino2020differences}, the following assumptions were made for activation rate $k_A$ and switching rate $f_m(\Tat)$:
\begin{eqnarray}
\label{cond1}
f_{m2}(\Tat) & = & 
\left\{ 
\begin{array}{ll}
\alpha_{m_2,A}/v_a & \mbox{if} \; \Tat < \Tat_{crit} \\
\alpha_{m_2,A} &  \mbox{if} \;  \Tat \ge \Tat_{crit}
\end{array} 
\right.\\
k_A & = & const
\label{cond2}
\end{eqnarray}

In this section we study stability of the corresponding linear dynamical system.

We can see that the system ~\eqref{eq:main}-\eqref{eq:main_end} allows for the zero steady state, which is always unstable assuming at least one of the LTRs is positive initially. This is due to the conservative property of LTRs, so if at least one of the LTR states starts out positive, all three states will be positive for all positive time, see Appendix~\ref{appendix:basic_properties}. 
This brings our focus to the the more interesting positive steady state of the system.
It is straightforward to show that all the eigenvalues are real for any values of the parameters, see Section~\ref{appendix:exact_sol_derived}. This motivates the following theorem.
\begin{theorem}
\label{theorem:gas_ori}
The system \eqref{eq:main}-\eqref{eq:main_end} under assumptions \eqref{cond1}-\eqref{cond2} has a unique asymptotically stable equilibrium, which is explicitly given in Appendix~\ref{appendix:steady_state_linear}.
\end{theorem}

We prove stability in two steps. First, we decouple the system and use Bendixon-Dulac criterion to show that the isolated (decoupled) system of $\LTR$s ($\LTR_R, \LTR_I, \LTR_A$) has a unique positive steady state that is globally asymptotically stable. It then follows directly that the entire system shares the same property.

\begin{proposition} The system \eqref{eq:main}-\eqref{eq:main3} of $\LTR$ (i.e. $\LTR_R, \LTR_I, \LTR_A$) under assumptions \eqref{cond1}-\eqref{cond2} has a unique positive steady state that is globally asymptotically stable.
\end{proposition}
\begin{proof} First, note that $\frac{d}{dt} \left[ \LTR_R + \LTR_I + \LTR_A \right] = 0$, so the system of $\LTR$s decoupled from the rest of the equations has a conservation law and the boundedness of $\LTR_R, \LTR_I, \LTR_A$ follows immediately. We consider the reduced two-dimensional system:
\begin{eqnarray}
\frac{d}{dt} \left[ \LTR_I \right] & = & k_{ON} - (k_A + k_{OFF} + k_{ON}) \LTR_I + (k_I - k_{ON}) \LTR_A \\
\frac{d}{dt} \left[ \LTR_A \right] & = & k_A \LTR_I - k_I \LTR_A
\end{eqnarray}

The nullclines of $\LTR_I$ and $\LTR_A$ are lines that intersect in the first quadrant, thus the system has a unique positive fixed point. Let $f(\cdot) = \frac{d}{dt} \left[ \LTR_I \right]$ and $g(\cdot) = \frac{d}{dt} \left[ \LTR_A \right]$. Observe that
\begin{equation}
\frac{\partial f}{\partial (\LTR_I)} + \frac{\partial g}{\partial (\LTR_A)} 
= - (k_A + k_{OFF} + k_{ON}) - k_I < 0. 
\end{equation}
By Bendixson-Dulac criterion, the system does not have a periodic orbit. Therefore the unique positive steady state is globally asymptotically stable, by Poincar\'e - Bendixson theorem. Furthermore, the unique positive steady state of the $\LTR$ system takes the forms
\begin{eqnarray}
\LTR_I^* & = & \frac{k_{ON} k_I}{(k_{OFF} + k_{ON}) k_I + k_{ON} k_A} \\
\LTR_A^* & = & \frac{k_{ON} k_A}{(k_{OFF} + k_{ON}) k_I + k_{ON} k_A} \\
\LTR_R^* & = & 1 - \LTR_I^* - \LTR_A^*
\end{eqnarray}

Stability of the original system follows by simply substituting equilibrium values of $\LTR$s into the rest of the equations. For instance, consider the rate equation for $\TAR$. Let $N_\LTR$ denote $\alpha_{m_1,R}\LTR_R^*+\alpha_{m_1,I}\LTR_I^* + \alpha_{m_1,A}\LTR_A^*$. Then in the limit as $t\rightarrow\infty$, the rate of change of $\TAR$ becomes:
\begin{equation}
    \frac{d}{dt} (\TAR) = N_{\LTR} - \gamma_{m_1} \TAR
\end{equation}
This implies the unique positive steady state $\TAR^* = N_{\LTR}/\gamma_{m_1}$ is globally stable. Similar argument holds for the remaining variables. Note that the steady state of $\env_A$ takes the form of a step function
\begin{eqnarray}
\env_A^* = 
\left\{ 
\begin{array}{ll}
\frac{\alpha_{m_2,A} \LTR_A^*}{v_a (\gamma_{m_2} + \alpha_{p_2}) } & \mbox{if} \; \Tat^* < \Tat_{crit} \\
\frac{\alpha_{m_2,A} \LTR_A^*}{\gamma_{m_2} + \alpha_{p_2} } &  \mbox{if} \;  \Tat^* \ge \Tat_{crit}
\end{array},
\right.
\end{eqnarray}
where $\Tat^* = (\alpha_{p_1} \env_I^*)/\gamma_{p_1}$.
\end{proof}
The existence of the globally asymptotically stable positive steady state of the system implies that the amount of virus would increase to a stable level once the activation of HIV-1 takes place.
However, such biological limit is difficult to meaningfully incorporate into any modeling schemes. Thus, this result should be interpreted as a \textit{stable} increase in the amount of virus after the activation of HIV-1.
In Appendix ~\ref{appendix:exact_sol_derived} we derive the complete closed form explicit steady state solution. In the case of this linear system we are able to describe dependence of the steady state on each of the parameters; however, a detailed examination of the sensitivity to each of the parameters is outside of the scope of this work. 
An advantage of having an explicit solution is the ability to estimate the rate of decay for each of the system variables using corresponding eigenvalues of the Jacobi matrix. 
This could prove useful for further analysis of the system and its potential modifications.

\subsection{Stability in case of $k_A$ and $f_m$ continuously depending on $\Tat$}
\label{ssec:Tat-dependent}
The model described above contains discontinuities due to the switching of the $f_m$ regimes based on $\Tat$ level, which may not be biologically valid. 
Furthermore, the switching function limits the ability of the model to distinguish the transcriptional dynamics as $\Tat$ approaches the critical threshold $\Tat_{crit}$ for T-cells and macrophages.
In what follows, we introduce a continuous version of the existing model by modifying the switches to the following form of Hill function:
\begin{eqnarray}
\label{eq:hill_func}
f_{m_2} (\Tat) & = & \frac{\alpha_{m_2,A}}{v_a} \frac{1 + v_a (Tat/T_c)^n}{1 + (Tat/T_c)^n}, \quad v_a>1 \\
\label{eq:hill_func_2}
k_A (\Tat) & = & \frac{\beta_{m_2,A}}{v_b} \frac{1 + v_b (Tat/T_c)^m}{1+ (Tat/T_c)^m}, \quad v_b>1 
\end{eqnarray}
Here $T_c$ denotes $\Tat_{crit}$. Figure~\ref{fig:hill_function} shows that  $f_m$ converges to the Heaviside step function as $n$ increases, asymptotycally approaching the form considered in the previous section, to characterize the $\Tat$-dependent rates in the transcription process. Not to impose additional assumptions, we reserve two different Hill constants $n$ and $m$ for $f_{m_2}(\Tat)$ and $k_A(\Tat)$. While $n$ often takes value between 2 and 3 in literature, or 1 in \cite{chavali2015distinct}, the possible biologically relevant ranges of $n$ and $m$ are all real numbers greater than or equal to 1.
\begin{figure}
    \centering
    \includegraphics[width=0.8\linewidth]{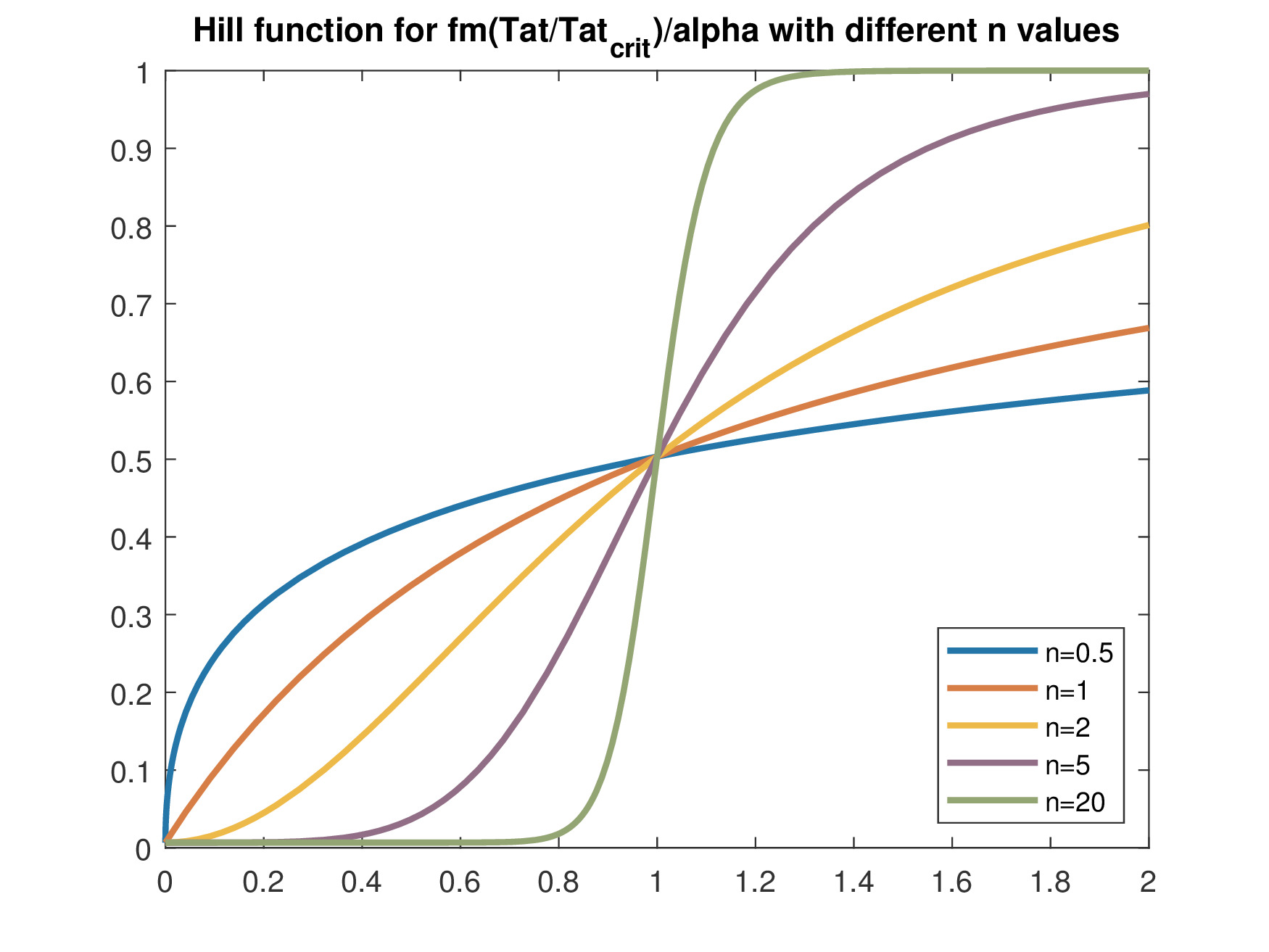}
    \caption{\footnotesize Hill function in Equation~\eqref{eq:hill_func} for different values of $n$. A similar observation is seen for the Hill function in Equation~\eqref{eq:hill_func_2}.}
    \label{fig:hill_function}
\end{figure}

We start examining the properties of the modified system by looking at its positive invariance. This is in line with our previous analysis since both functional responses are positive and bounded.
\begin{lemma} The system \eqref{eq:main}-\eqref{eq:main_end} under assumptions \eqref{eq:hill_func}-\eqref{eq:hill_func_2} is positively invariant.
\end{lemma}
\begin{proof}
 Since $f_{m_2}(\Tat)$ and $k_A(\Tat)$ are bounded above by $\alpha_{m_2,A}$ and $\beta_{m_2,A}$ and $v_a$ and $v_b$ are taken to be strictly greater than 1, the positive invariance of the new system follows directly from the boundedness of the original system.
\end{proof}
The introduction of the continuous functional responses allow for more interesting coupling of the dynamics between different variables; however, the overall dynamics of the system still does not rely on TAR, Pr55 and p24, since they only feed on the other six variables.
Additionally, within the remaining six equations, $\env_A$ does not contribute to the dynamics of the other five and $\LTR$ is conserved. 
Thus, we start our analysis on the reduced system of four differential equations \eqref{eq: LTR_I},\eqref{eq: LTR_A},\eqref{eq: Tat},\eqref{eq: env_I}.

Define $x = \LTR_I, y = \LTR_A, s = \env_I, v = \Tat, a_1 = k_{ON}, a_2 = k_I - k_{ON}, a_3 = k_{ON} + k_{OFF}, a_4 = k_I, a_5 = \alpha_{p_1}, a_6 = \gamma_{p_1}, a_7 = \alpha_{m_2,I}, a_8 = \gamma_{m_2} + \alpha_{p_1} + \alpha_{p_2}, a_9 = \beta_{m_2,A}/v_b$, $a_{10} = v_b/(T_c)^n$ and $a_{11} = 1/(T_c)^n$. Note that $a_{10} = v_b a_{11}$, so since we take $v_b$ to be strictly larger than 1, $a_{10} > a_{11}$.
In these notations, the system takes on the form:

\begin{eqnarray}
\label{eq:reduced_1}
x' & = & a_1 + (a_4 - a_1) y - \left( a_3 + a_9 \frac{1 + a_{10} v^n}{1 + a_{11} v^n} \right) x \\
y' & = & a_9 \frac{1 + a_{10} v^n}{1 + a_{11} v^n} x - a_4 y\\
s' & = & a_7 x - a_8 s\\
\label{eq:reduced_2}
v' & = & a_5 s - a_6 v.
\end{eqnarray}
As with the linear system, the nonlinear system also contains a zero steady state that is always unstable whenever at least one of the initial conditions for LTR is positive. Thus, we focus our analysis on the positive steady state. Note that for $n < 1$, the Hill function exhibits dynamics that are unexpected for our biological system, see Figure~\ref{fig:hill_function}, so we discard that case.
\begin{proposition}
 The reduced system, equation~\eqref{eq:reduced_1}-\eqref{eq:reduced_2}, has a unique positive steady state for all $n \geq 1$.
\end{proposition}
\begin{proof} 
 In order to show the system has a unique positive steady state, we proceed by finding the nullclines. Setting $v' = 0$ and $s' = 0$, we obtain:
\begin{eqnarray}
x^* & = & \frac{a_8}{a_7} s^*\\
s^* & = &\frac{a_6}{a_5} v^*.
\end{eqnarray}
Together, this gives:
\begin{eqnarray}
v^* & = & \frac{a_5 a_7}{a_6 a_8} x^* =: \bar{w} x^*,
\end{eqnarray}
where $\bar{w} = \frac{a_5 a_7}{a_6 a_8}$. Next we set $y' = 0$ and solve for $y^*$ in term of $x^*$ to obtain:
\begin{eqnarray}
y^* & = & \frac{a_9}{a_4} \frac{1 + a_{10} (v^*)^n}{1 + a_{11} (v^*)^n} x^*\\
  & = & \frac{a_9}{a_4} \frac{1 + a_{10} (\bar{w} x^*)^n}{1 + a_{11} (\bar{w} x^*)^n} x^*.
\end{eqnarray}
Finally, setting $x' = 0$ and replace in $y^*$.
\begin{eqnarray}
a_1 + \frac{a_9}{a_4}(a_4 - a_1) \frac{1 + a_{10} (\bar{w} x^*)^n}{1 + a_{11} (\bar{w} x^*)^n} x^* \nonumber\\
- \left( a_3 + a_9 \frac{1 + a_{10} (\bar{w} x^*)^n}{1 + a_{11} (\bar{w} x^*)^n} \right) x^* & = & 0.
\end{eqnarray}
Rearranging terms, we obtain:
\begin{eqnarray}
a_1 - a_3 x^* & = & \frac{a_1 a_9}{a_4} \left( \frac{1 + a_{10} (\bar{w} x^*)^n}{1 + a_{11} (\bar{w} x^*)^n} \right) x^*.
\end{eqnarray}
Letting $g(x^*) := a_1 - a_3 x^*$ and $f(x^*) := \frac{a_1 a_9}{a_4} \left( \frac{1 + a_{10} (\bar{w} x^*)^n}{1 + a_{11} (\bar{w} x^*)^n} \right) x^*$. Note that both $f(x^*)$ and $g(x^*)$ are continuous and strictly monotone functions on $x^* \in [0, 1]$. Furthermore, $f(0) = 0$ and $f(1) > 0$, while $g(0) > 0$ and $g(1) = a_1 - a_3 = - k_{OFF} < 0$. Thus invoking the intermediate value theorem, we have $f(x^*)$ and $g(x^*)$ intersect at a unique point $x^* \in (0,1)$. It follows immediately the system has a unique positive equilibrium $(x^*, y^*, s^*, v^*)$.
\end{proof}

The positivity of the unique nontrivial steady state of our system helps establish its biological validity. However, it does not rule out the possibility of finding a steady state arbitrarily close to 0, which is unrealistic in practice when HIV-1 viral load stays low but away from 0. For this reason, we establish a proposition that establishes lower bounds on all system variables.

Recall that we say the system described in Equations (\ref{eq:reduced_1}-\ref{eq:reduced_2}) is permanent 
if there are positive constant M and N such that
\begin{equation}
    \limsup_{t\rightarrow\infty} \max\{x(t),y(t),s(t),v(t)\} < M,
\end{equation}
and if
\begin{equation}
    \liminf_{t\rightarrow\infty} \min\{x(t),y(t),s(t),v(t)\} > N.
\end{equation}

We will show that the system in Equations (\ref{eq:reduced_1}-\ref{eq:reduced_2}) is permanent in the above sense for $n\geq1$. By construction, $x$ and $y$ are bounded above by 1, so we only need to show they also have a positive lower bound.

\begin{proposition}
\label{prop:lower_bound_x_y}
There exist positive constants $m_x$ and $m_y$ such that $0 < m_x \leq \liminf_{t\rightarrow\infty} x$ and $0 < m_y \leq \liminf_{t\rightarrow\infty} y$. 
\end{proposition}
\begin{proof} First we note that all variables in Equations (\ref{eq:reduced_1}-\ref{eq:reduced_2}) are non-negative. 
Thus, for all non-negative values of $x, y, v$, we consider a new variable $\underbar{X}(t)$ such that $\underbar{X}(0) = x(0)$ and
\begin{eqnarray}
    \underbar{X}' & = & \min\left\{a_1,a_4\right\} - \left( a_3 + a_9 v_b \right) \underbar{X}.
\end{eqnarray}
Consider $x'  =  a_1 + (a_4 - a_1) y - \left( a_3 + a_9 \frac{1 + a_{10} v^n}{1 + a_{11} v^n}\right)x$.
Observe that if $a_4>a_1$, then $\min\{a_1 + (a_4 - a_1)y\} \geq a_1$, with minimum achieved when $y=0$. 
Otherwise, $a_4<a_1$, then $\min\{a_1 + (a_4-a_1)y\} \geq a_4$ with minimum achieved when $y=1$. 
Furthermore, observe that
\begin{eqnarray}
 a_3 + a_9 \frac{1 + a_{10} v^n}{1 + a_{11}v^n} & \leq & 
 a_3 + a_9 \max\left\{\frac{a_{10}}{a_{11}},1\right\} \\
 & = & a_3 + a_9 \max\{v_b,1\} \\
 & = & a_3 + a_9 v_b,
\end{eqnarray}
where the last equality follows from the assumption that the amplification effect due to activation, $v_b$, is strictly greater than 1.

Now, we claim that $x(t) \geq \underbar{X}(t)$ for all $t\geq 0$. If not, then since $x(t)$ and $\underbar{X}(t)$ are non-negative, $x(0) = \underbar{X}(0)$ and $x'(0) \geq \underbar{X}'(0)$, there exists $t_1 > 0$ such that $x(t) \geq \underbar{X}(t)$ for $t \in [0,t_1)$ and $x(t_1)= \underbar{X}(t_1)$ with $x'(t_1) < \underbar{X}'(t_1)$. However, we note that
\begin{eqnarray}
    x'(t_1) & = & a_1 + (a_4 - a_1) y(t_1) - \left( a_3 + a_9 \frac{1 + a_{10} v^n(t_1)}{1 + a_{11} v^n(t_1)}\right)x(t_1)\\
         & \geq & \min\left\{a_1,a_4\right\} - \left( a_3 + a_9 v_b \right) x(t_1) \\
            & = & \min\left\{a_1,a_4\right\} - \left( a_3 + a_9 v_b \right) \underbar{X}(t_1)\\
            & = & \underbar{X}'(t_1),
\end{eqnarray}
which is a contradiction. Hence, $x(t) \geq \underbar{X}(t)$ for all $t\geq 0$.
Observe that since $\lim_{t\rightarrow\infty} \underbar{X}(t) = \frac{\min\{a_1,a_4\}}{ a_3 + a_9 v_b }$, this implies 
\begin{eqnarray}
\liminf_{t\rightarrow\infty} x(t) \geq \lim_{t\rightarrow\infty} \underbar{X}(t) = \frac{\min\{a_1,a_4\}}{ a_3 + a_9 \max\{a_{10},1\} } > 0 .
\end{eqnarray}
Define $m_x = \frac{\min\{a_1,a_4\}}{a_3 +a_9 v_b}$, then eventually $x(t) \geq m_x > 0$. 

%

Consider $\underbar{Y}(t)$ such that $\underbar{Y}(0) = y(0)$ and
\begin{eqnarray}
\underbar{Y}' & = & a_9 m_x - a_4 \underbar{Y}.
\end{eqnarray}
Similarly, we obtain that $y(t) \geq \underbar{Y}(t)$ for $t\geq 0$ and
\begin{eqnarray}
\liminf_{t\rightarrow\infty} y(t) \geq \lim_{t\rightarrow\infty} \underbar{Y}(t) = \frac{a_9 m_x}{a_4} > 0.
\end{eqnarray}
Define $m_y = \frac{a_9 m_x}{a_4}$. This concludes our proof.

%
\end{proof}
From Proposition~\ref{prop:lower_bound_x_y}, it is straightforward to show that $s$ and $v$ also have positive lower and upper bound. Thus we state the following Lemma without proof.
\begin{lemma}
\label{lem:permanent}
The system in Equations (\ref{eq:reduced_1}-\ref{eq:reduced_2}) is permanent.
\end{lemma}
Note that Lemma~\ref{lem:permanent} also guarantees that the nonlinear system with $k_A(\Tat)$ and $f_m(\Tat)$ continuously dependent on Tat is also permanent.

\subsection{Two alternative 3-dimensional approximations}
\label{ssec:reduced_models}
%
Even with the reduction, the asymptotic dynamics of the reduced system ~\eqref{eq:reduced_1}-\eqref{eq:reduced_2} is still difficult to study.
Thus, we examine two alternative models that capture the asymptotic behavior of the reduced system. 

First, since we often observe that $\alpha_{m_2,I} \gg \gamma_{m_2} + \alpha_{p_1} + \alpha_{p_2}$ (see Table~\ref{table:nonlinear_params}), it implies $a_7 \gg a_8$.
Thus, one may consider the quasi-steady state (QSS) approximation that $s(t) \approx \frac{a_7}{a_8}x(t)$ or $\env_I \approx \frac{a_7}{a_8}\LTR_I$.

This gives the following {\bf QSS system}:
\begin{eqnarray}
\label{eq: qss1}
x' & = & a_1 + (a_4 - a_1) y - \left( a_3 + a_9 \frac{1 + a_{10} v^n}{1 + a_{11} v^n} \right) x \\
y' & = & a_9 \frac{1 + a_{10} v^n}{1 + a_{11} v^n} x - a_4 y\\
v' & = & a_5 \frac{a_7}{a_8}x - a_6 v.
\label{eq: qss3}
\end{eqnarray}

Alternatively, one can treat the $s$ compartment as a delay factor in the link between $x$ and $v$. In other words, we assume $s(t) \approx \frac{a_7}{a_8}x(t-\tau)$, where $\tau$ is a pre-determined time delay ($\tau \approx 1/a_8$). This leads to the following {\bf Delay system}:
\begin{eqnarray}
\label{eq: dm1}
x' & = & a_1 + (a_4 - a_1) y - \left( a_3 + a_9 \frac{1 + a_{10} v^n}{1 + a_{11} v^n} \right) x \\
y' & = & a_9 \frac{1 + a_{10} v^n}{1 + a_{11} v^n} x - a_4 y\\
v' & = & a_5 \frac{a_7}{a_8}x(t-\tau) - a_6 v.
\label{eq: dm3}
\end{eqnarray}
In Figure~\ref{fig:Model comparison T-cells} we provide computational comparison of these two models against the original one. Both approximations are able to capture asymptotic behavior of the system, but there are noticeable differences in transient dynamics, as expected. 

%
        
%
\begin{figure}[!ht]
		\centering
\includegraphics[width=1\linewidth]{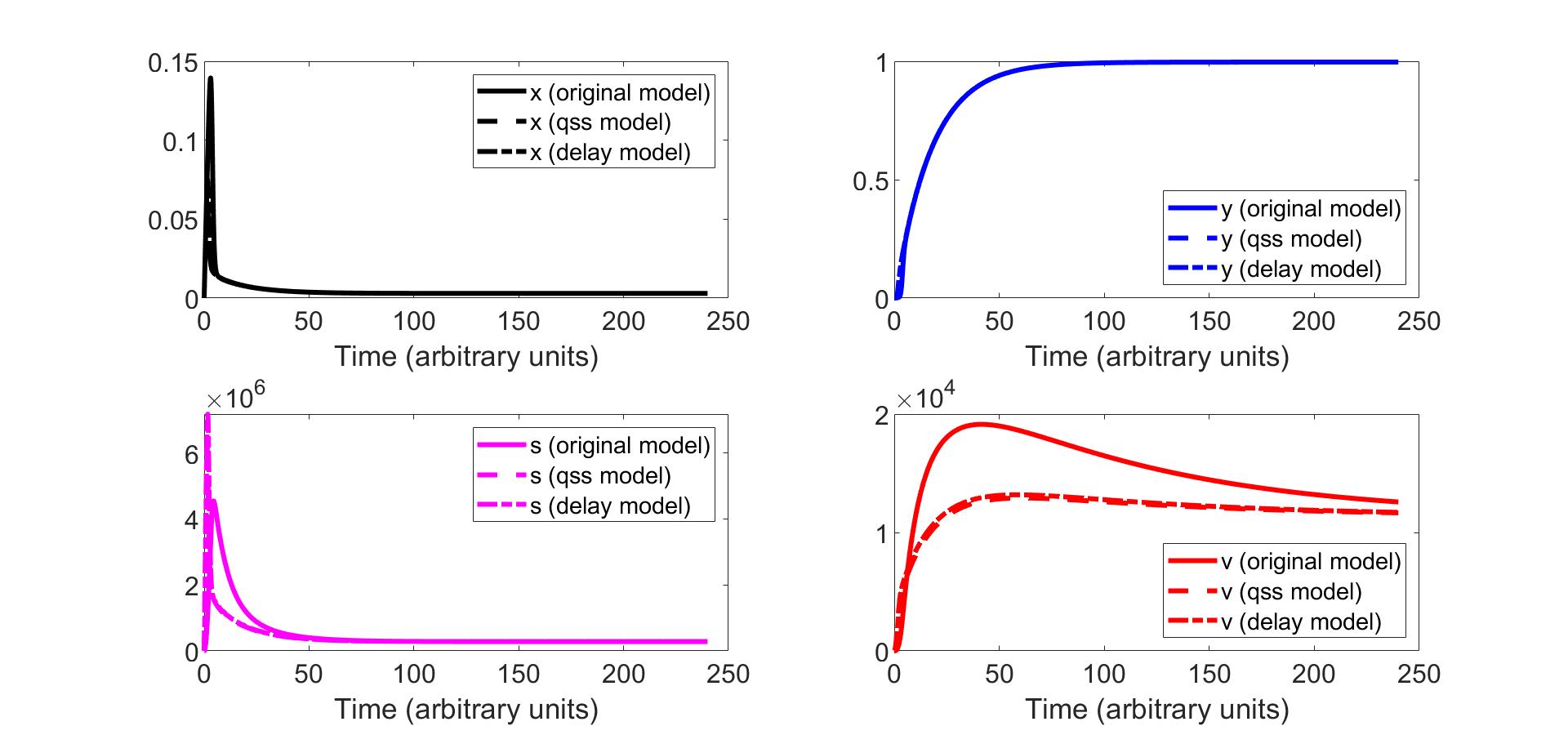}
        \caption{\footnotesize Numerical comparison between the 4-dimensional model ~\eqref{eq:reduced_1}-\eqref{eq:reduced_2}, QSS model \eqref{eq: qss1}-\eqref{eq: qss3} and Delayed model \eqref{eq: dm1}-\eqref{eq: dm3}. The y-axis represents the arbitrary unit of the variable. Parameter values are for T-cells. While all models show similar behavior for $x$ and $y$ dynamics, reduced models underestimate the growth of $v$ variable and there are differences in transient behavior for the $s$ variable.}
        \label{fig:Model comparison T-cells}
	\end{figure}
Since the Delayed model does not offer significant advantages over the QSS model in terms of capturing long term system behavior based on this calculation, in our stability analysis we will focus on the QSS model. We note that the Delayed model might be useful in case a more careful analysis of the transient mode is of interest.

\subsection{Stability of the QSS model}
%
For the quasi-steady state model \eqref{eq: qss1}-\eqref{eq: qss3}, its boundedness, positive invariance and the existence of a unique positive steady state are direct consequences of the results we established earlier for the original 4-dimensional model. Thus, we only need show local asymptotic stability for the positive steady state.
\begin{proposition} 
The positive steady state of the QSS system \eqref{eq: qss1}-\eqref{eq: qss3} is locally asymptotically stable.
\end{proposition}
\begin{proof}
We consider the Jacobian of the system at $(x^*, y^*, v^*)$,
\begin{eqnarray}
J (x^*, y^*, v^*) =
\begin{pmatrix}
  -\left(a_3 + a_9 \frac{1 + a_{10} (v^*)^n}{1 + a_{11}(v^*)^n} \right) & a_4 - a_1 & -a_9 x^* \frac{(a_{10} - a_{11})n(v^*)^{n-1}}{(1+a_{11}(v^*)^n)^2}\\
  a_9 \frac{1 + a_{10} (v^*)^n}{1 + a_{11}(v^*)^n} & -a_4 & a_9 x^* \frac{(a_{10} - a_{11})n(v^*)^{n-1}}{(1+a_{11}(v^*)^n)^2} \\
 a_5\frac{a_7}{a_8} & 0 & -a_6 
 \end{pmatrix}.
\end{eqnarray}
Denote $\Delta_1 := a_9 \frac{1 + a_{10} (v^*)^n}{1 + a_{11}(v^*)^n}$, 
$\Delta_2 = a_9 x^* \frac{(a_{10} - a_{11})n(v^*)^{n-1}}{(1+a_{11}(v^*)^n)^2}$ and $\alpha = a_5 \frac{a_7}{a_8}$. Then the Jacobian matrix becomes:
\begin{eqnarray}
J (x^*, y^*, v^*) =
\begin{pmatrix}
  -\left(a_3 + \Delta_1 \right) & a_4 - a_1 & -\Delta_2\\
  \Delta_1 & -a_4 & \Delta_2 \\
 \alpha & 0 & -a_6 
 \end{pmatrix}.
\end{eqnarray}
Solving for the determinant of $J (x^*, y^*, v^*)-I \lambda$ to obtain the characteristic equation,
\begin{eqnarray}
\det( J (x^*, y^*, v^*)-I \lambda) & = &
\begin{pmatrix}
  -\left(a_3 + \Delta_1 \right)-\lambda & a_4 - a_1 & -\Delta_2\\
  \Delta_1 & -a_4-\lambda & \Delta_2 \\
 \alpha & 0 & -a_6 - \lambda
 \end{pmatrix}\\
& = & -(a_6 + \lambda)[(a_4+\lambda)(a_3+\Delta_1+\lambda) \nonumber\\
& & -\Delta_1(a_4-a_1)] \nonumber\\
&   & + \alpha[\Delta_2(a_4-a_1) - \Delta_2(a_4 + \lambda)]\\
& = & -(a_3 a_4 a_6 + \Delta_1 a_1 a_6 + \alpha a_1 \Delta_2) \nonumber\\
&   & -\lambda(a_4 a_6 + a_3 (a_4 + a_6) \nonumber\\
& & + \Delta_1(a_6+a_1)+\alpha \Delta_2)\nonumber\\
&   & -\lambda^2(a_4+a_6+a_3+\Delta_1) - \lambda^3.
\end{eqnarray}
Thus the characteristic equation takes the form,
\begin{eqnarray}
 q(\lambda) & = & \lambda^3 + \lambda^2(a_4+a_6+a_3+\Delta_1) + \lambda(a_4 a_6 + a_3 (a_4 + a_6)+\nonumber\\
 &   & + \Delta_1(a_6+a_1)+\alpha \Delta_2) + (a_3 a_4 a_6 + \Delta_1 a_1 a_6 + \alpha a_1 \Delta_2)\\
 & =: & \lambda^3 + \lambda^2 A_2 + \lambda A_1 + A_0.
\end{eqnarray}
Here $A_0=a_3 a_4 a_6 + \Delta_1 a_1 a_6 + \alpha a_1 \Delta_2$, $A_1=a_4 a_6 + a_3 (a_4 + a_6)+\Delta_1(a_6+a_1)+\alpha \Delta_2$ and $A_2=a_4+a_6+a_3+\Delta_1$.
By Routh-Hurwitz criterion, the condition for stability is satisfied if: (1) $A_2, A_0 > 0$, and (2) $A_2 \cdot A_1 > A_0$.

Since all parameters are positive, the condition (1) is satisfied. Additionally,
\begin{eqnarray}
 A_2\cdot A_1 - A_0 & = & (a_4 + a_6 + a_3 + \Delta_1) \nonumber\\
 & & \times(a_4 a_6 + a_3 a_4 + a_3 a_6 + a_6 \Delta_1 + a_1 \Delta_1 + \alpha \Delta_2)\nonumber\\
 &   & - (a_3 a_4 a_6 + a_1 a_6 \Delta_1 + \alpha a_1 \Delta_2)\\
 & = & (a_4+a_6+a_3+\Delta_1)(a_3a_4 + a_1\Delta_1 +\alpha \Delta_2)\nonumber\\
 &   & + (a_4+a_6+a_3+\Delta_1)(a_4a_6 + a_3a_6+a_6a_1)\nonumber\\
 &   & - (a_3 a_4 a_6 + a_1 a_6 \Delta_1 + \alpha a_1 \Delta_2)\\
 & = & (a_4+a_3+\Delta_1)(a_3a_4 + a_1\Delta_1 +\alpha \Delta_2)\nonumber\\
 &   & + (a_3 - a_1)\alpha \Delta_2 + (a_4+a_6+\Delta_1)\alpha \Delta_2\nonumber\\
 &   & + (a_4+a_6+a_3+\Delta_1)(a_4a_6 + a_3a_6+a_6a_1).
\end{eqnarray}
Since $a_3 = a_1 + (\text{a positive number})$ by definition, $A_2A_1>A_0$. Hence, the second condition is also satisfied.
\end{proof} 
The complete global stability result is difficult to obtain even for the quasi-steady state system (using standard Lyapunov functions). 
Instead, we observe that $\Tat_{crit}$ is several orders of magnitude smaller than the value of $\Tat$ shortly after the experiment starts. This means the nonlinear model is quickly reduced to the linear model, unless $\Tat_{crit}$ is significantly larger.
Qualitatively, this means that if the production rate of $\Tat$ is high enough, then we can expect the new model to show similar dynamical behavior to the original system (e.g. the positive steady state is globally asymptotically stable).
A similar observation should hold for very low production rate of $\Tat$.
%

In Fig.~\ref{fig:bifurcation} we numerically study the dependence of steady state on varying parameters over a reasonable range. The following results are representative of the study. They show that under reasonable parameter ranges, the positive steady state is always stable.
These observations suggest that the unique fixed point is expected to be global stable for the nonlinear system. 
%


%
\begin{figure}[!ht]
		\centering
\includegraphics[width=1\linewidth]{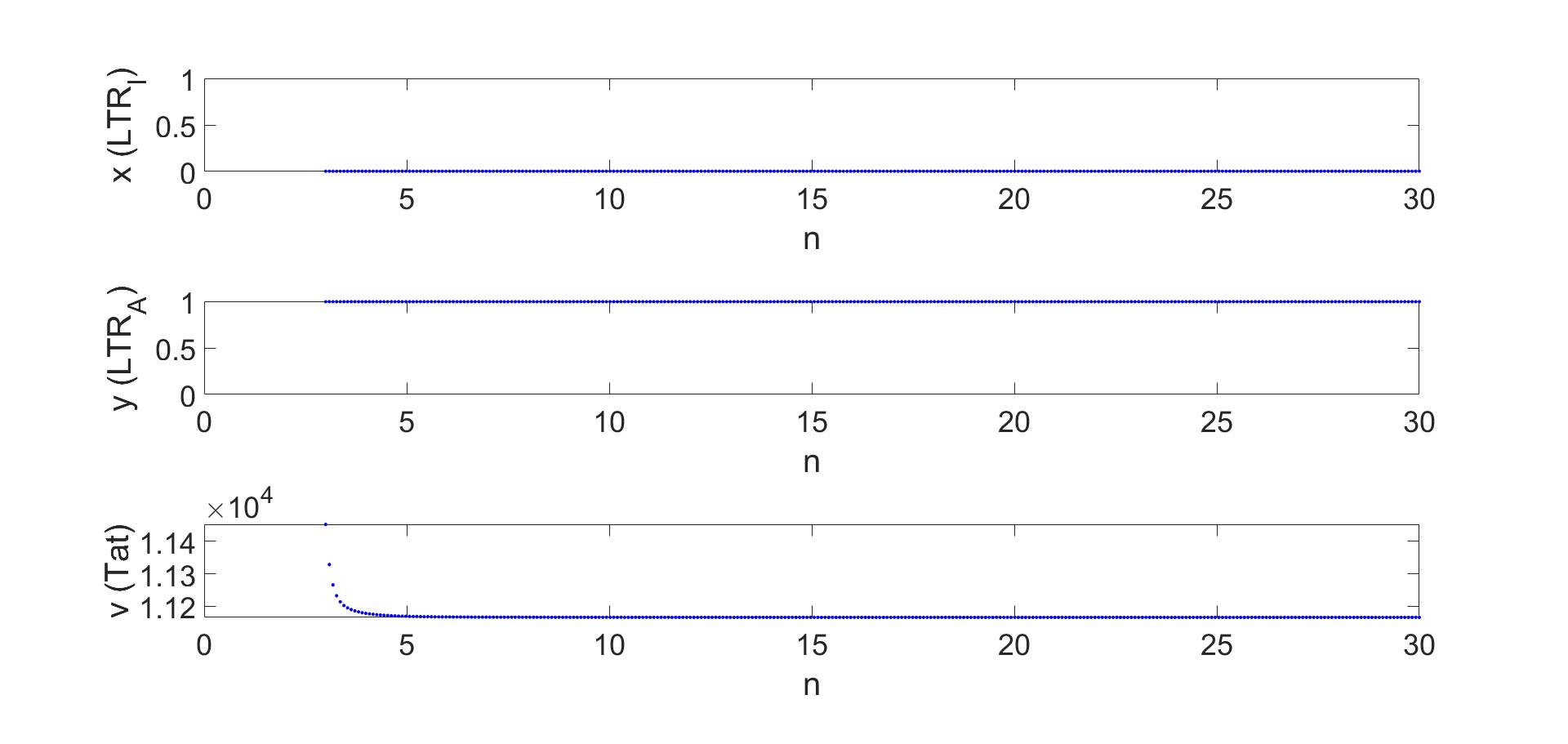}
        \caption{\footnotesize Numerical results for estimating the effect of varying parameter $n$ on the behavior of the QSS system \eqref{eq: qss1}-\eqref{eq: qss3}. Note that the positive steady state for $x$ is very small, but positive.}
        \label{fig:bifurcation}
	\end{figure}
%
\section{Parameter estimation - comparison of linear and nonlinear models}
\label{sec:param_esti}
%
%
In \cite{demarino2020differences}, we collected time series data for TAR and total $\env$ RNA. In that work, standard least squares method was used to fit the {\it linear model} given by \eqref{eq:main}-\eqref{eq:main_end} under assumptions \eqref{cond1}-\eqref{cond2} to all the data points simultaneously. We utilize the same method for direct comparison with the behavior of the {\it nonlinear model} specified by \eqref{eq:main}-\eqref{eq:main_end} under assumptions \eqref{eq:hill_func}-\eqref{eq:hill_func_2}. The fitting parameters are given by $\alpha_{m_1,A}, \alpha_{m_2,A}, \gamma_{m_1}$ and $\gamma_{m_2}$. While it is also possible to fit other parameters (especially the degradation of $\Tat$), this objective will likely result in over-fitting due to limited data. Our main purpose is to compare the values of numerically estimated parameters between linear and nonlinear models.

The function {\tt fmincon} in MATLAB is used to estimate these parameters within the same ranges as described in \cite{demarino2020differences}. Additionally, the range for $v_b$ is taken to be $[1,200]$ (i.e. around the value of $v_a$) and the ranges for $n$ and $m$ are taken to be $[1, \infty)$ since there are no known biological constraints for their upper bound.
The estimated values are presented in Table~\ref{table:nonlinear_params}.
\begin{figure}[!ht]
\centering
\includegraphics[width=1\linewidth]{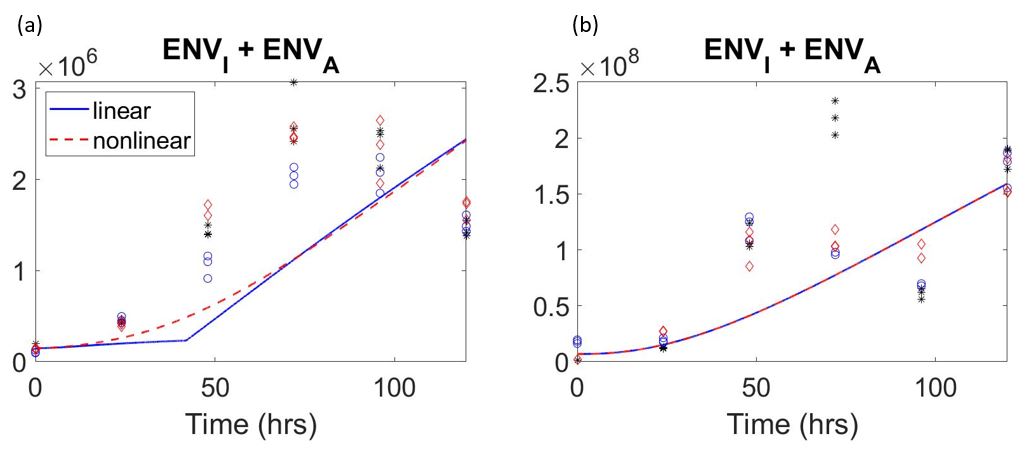}
\caption{\footnotesize Comparison of linear and nonlinear models fitted to \env~data in both T-cells and macrophages. (a) \env~data and the fitted models behavior in macrophages. (b) \env~data and the fitted models behavior in T-cells.}
        \label{fig:env_fitting}
\end{figure}
\begin{figure}[!ht]
\centering
\includegraphics[width=1\linewidth]{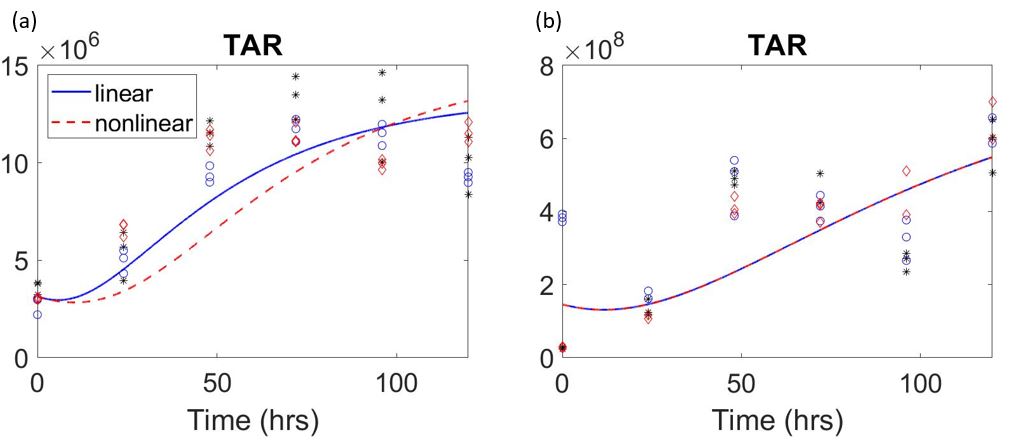}
\caption{\footnotesize Comparison of linear and nonlinear models fitted to \TAR~data in both T-cells and macrophages. (a) \TAR~data and the fitted models behavior in macrophages. (b) \TAR~data and the fitted models behavior in T-cells.}
        \label{fig:tar_fitting}
\end{figure}
\begin{figure}[!ht]
\centering
\includegraphics[width=1\linewidth]{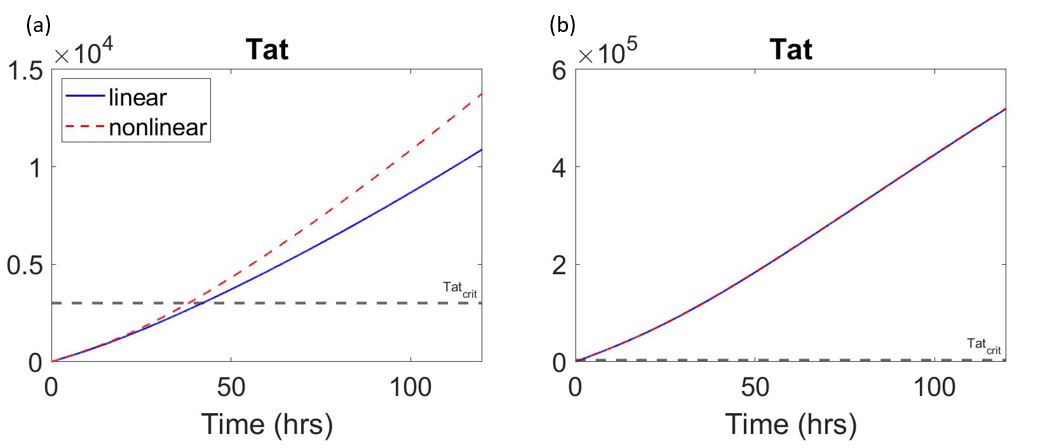}
\caption{\footnotesize Comparison of the dynamics of Tat resulted from fitting the linear and nonlinear models \TAR~data and \env~data in both T-cells and macrophages. (a) \Tat~dynamics in macrophages. (b) \Tat~dynamics in T-cells.}
        \label{fig:tat_fitting}
\end{figure}

Figures~\ref{fig:env_fitting}-\ref{fig:tar_fitting} show that both models produce similar fits for both types of cells. Dynamics for other variables are also very similar.
This is likely because even though the effect of Tat is significant in the dynamics of the model, the values of Tat in the linear and non-linear model remain relatively similar before and immediately after reaching the Tat$_{crit}$ threshold, see Figure~\ref{fig:tat_fitting}(a). 
Alternatively, it is also possible that the parameter Tat$_{crit}$ is orders of magnitude smaller than the value of Tat after several hours, see Figure~\ref{fig:tat_fitting}(b). In both scenarios, the effect of Tat on the dynamics of the system quickly saturates in a similar manner in both models. However, the nonlinear model gives a noticeable difference in the Tat dynamics as compared to the linear model in the case of macrophages. 
Yet, both models produce similar dynamics for $\Tat$ in T-cells. This observation poses an interesting modeling question (outside of the scope of this paper) that may also be biologically relevant: why are the $\Tat$ dynamics predicted by both models only differ only in the case of macrophages?

\begin{table}
\begin{center}
\begin{tabular}{lcccc}
\hline
                 & Unit                              & definition                  & T-Cell    & macrophages \\ \hline\hline
$\beta_{m_2,A}$    & [change/mL/hr] & $\LTR_I \rightarrow \LTR_A$ & same as $k_{A}(\Tat)$ & same as $k_{A}(\Tat)$ \\[.75 ex]
$\alpha_{m_1,A}$ & [copies/mL/hr]   & $\LTR_A \rightarrow \TAR$   & $1.38 \times 10^7$   & $5.25 \times 10^5$ \\[.75 ex]
$\alpha_{m_2,A}$ & [copies/mL/hr]   & $\LTR_A \rightarrow \env_I$ & $2.47 \times 10^6$   & $5.45 \times 10^4$ \\[.75 ex]
$\gamma_{m,1}$   & [$\TAR$/mL/hr]   & $\TAR$ degradation   & $1.17 \times 10^4$    & $2.68 \times 10^4$ \\[.75 ex]
$\gamma_{m,2}$   & [$\env_I$/mL/hr] & $\env_I$ degradation & $2.25 \times 10^3$    & $2.91 \times 10^4$ \\[.75 ex]
$v_b$            & unitless         & folds of increasing  & $10$  & $10$ \\ [.75 ex]
$n$            & unitless         & degree of non-linearity  & $6.23$  & $1$ \\ [.75 ex]
$m$            & unitless         & degree of non-linearity  & $4.98 \times 10^1$  & $1$ \\ [.75 ex]
 \hline
 \end{tabular} \\
 \end{center}
\vspace{0.0in}
\caption{Values of additional parameters for the nonlinear model. Parameters not mentioned here are taken to be the same as their values in Table~\ref{table:parameters}.}
\label{table:nonlinear_params}
\end{table}





\section{Effect of different drug types}
\label{sec:drug_effect}
\subsection{The effect of drug type F07\#13 in combination with standard treatments of HIV}
Concerning the drug F07\#13, when it is administered, the values of $w_1, w_5$ increase higher than 1. This leads to an increase in the steady states $\LTR_I^*$ and $\LTR_R^*$, while $\LTR_A^*$ will decrease. Consequently, the value of $\env_I^*$ will increase, while $\env_A^*$ will decrease. 
These effects eventually affect the production of $\Pr55$.
Since $\Pr55$ can be used as a tracker for viral proteins production, the effect of F07\#13 may potentially be studied by looking at how it affects the dynamics of $\Pr55$.

%
%

First we will demonstrate that it is not trivial that treatment using F07\#13 will reduce the production of $\Pr55$.
Note that F07\#13 reduces the activation rate of $\LTR$, so it negatively affects the proportion of $\LTR_A$ and consequentially the production of $\env_A$ and the corresponding production of $\Pr55$ from $\env_A$.
However, this comes at a cost of increasing the proportion of $\LTR_I$, which increases the production of $\env_I$ that also contributes to the production of $\Pr55$ at an equal rate to that of $\env_A$.
The contributions of $\env_I$ and $\env_A$ are reflected in the final state $\Pr55^*$ at an equal rate of $\frac{\alpha_{p_2} w_2}{\alpha_{p_3}}$. 
In other words, it is not obvious whether or not F07\#13 effectively decreases $\Pr55$.

For instance, consider their ratio when $w_2 = w_3 = w_4 = 1$ (only F07\#13 is present) and $\Tat^* \geq \Tat_{crit}$:
\begin{eqnarray}
\frac{\env_I^*}{\env_A^*} & = & w_1 w_5 \frac{\alpha_{m_2,I}}{\alpha_{m_2,A}} \frac{\gamma_{m_2} + \alpha_{p_2}}{\gamma_{m_2} + \alpha_{p_1} + \alpha_{p_2}} \frac{k_I}{k_A}.
\end{eqnarray}
This shows that F07\#13 affects the relative concentration of $\env_I^*$ and $\env_A^*$; however, the actual amount of increasing/decreasing due to F07\#13 is not clear. 

For the reasons mentioned above we take a different approach. We let $w := \frac{1}{w_1 w_5}$ and rewrite $\Pr55^*$ in term of $w$:
\begin{align}
\Pr55^*(w) &= \left(\frac{\alpha_{p_2}}{\alpha_{p_3}} \right) \left[ \frac{k_{ON} k_I}{k_I (k_{ON} + k_{OFF}) + k_{ON} k_A w }\right] \nonumber\\ 
& \times \frac{\alpha_{m_2,I}}{\gamma_{m_2} + \alpha_{p_1} + \alpha_{p_2}} \nonumber\\
& + \left(\frac{\alpha_{p_2}}{\alpha_{p_3}} \right) \left[ \frac{k_{ON} k_A}{k_I (k_{ON} + k_{OFF}) + k_{ON} k_A w }\right] \frac{\alpha_{m_2,A}}{\gamma_{m_2} + \alpha_{p_2}} w \\
&= \left(\frac{\alpha_{p_2}}{\alpha_{p_3}} \right) \left[ \frac{k_{ON}}{k_I (k_{ON} + k_{OFF}) + k_{ON} k_A w }\right] \nonumber\\
& \times \left(\frac{k_I \alpha_{m_2,I}}{\gamma_{m_2} + \alpha_{p_1} + \alpha_{p_2}} + \frac{k_A \alpha_{m_2,A}}{\gamma_{m_2} + \alpha_{p_2}} w\right).
\end{align}
Taking derivative with respect to $w$,
\begin{align}
\frac{d}{dw} \Pr55^*(w) &= \left(\frac{\alpha_{p_2}}{\alpha_{p_3}} \right) \left[-  \frac{(k_{ON}) (k_{ON} k_A)}{(k_I (k_{ON} + k_{OFF}) + k_{ON} k_A w)^2 }\right] \nonumber\\
& \times \left(\frac{k_I \alpha_{m_2,I}}{\gamma_{m_2} + \alpha_{p_1} + \alpha_{p_2}} + \frac{k_A \alpha_{m_2,A}}{\gamma_{m_2} + \alpha_{p_2}} w\right) \nonumber\\
 &+ \left(\frac{\alpha_{p_2}}{\alpha_{p_3}} \right) \left[ \frac{k_{ON}}{k_I (k_{ON} + k_{OFF}) + k_{ON} k_A w }\right] \frac{k_A \alpha_{m_2,A}}{\gamma_{m_2} + \alpha_{p_2}}\\
 &= \left(\frac{\alpha_{p_2}}{\alpha_{p_3}} \right) \left[ \frac{k_{ON} }{k_I (k_{ON} + k_{OFF}) + k_{ON} k_A w }\right] \left( \frac{k_A \alpha_{m_2,A}}{\gamma_{m_2} + \alpha_{p_2}} \right) \nonumber\\
 & \left[1 - \frac{k_{ON} k_A}{k_I (k_{ON} + k_{OFF}) + k_{ON} k_A w }
 \times \left( w + \frac{\alpha_{m_2,I}}{\alpha_{m_2,A}} \frac{\gamma_{m_2} + \alpha_{p_2}}{\gamma_{m_2} + \alpha_{p_1} + \alpha_{p_2}} \frac{k_I}{k_A} \right) \right].
\end{align}
Now at $w = 1$ (or with no F07\#13 drug), then
\begin{align}
\frac{d}{dw} \Pr55^*(w = 1) &= \left(\frac{\alpha_{p_2}}{\alpha_{p_3}} \right) \left[ \frac{k_{ON} }{k_I (k_{ON} + k_{OFF}) + k_{ON} k_A }\right] \left( \frac{k_A \alpha_{m_2,A}}{\gamma_{m_2} + \alpha_{p_2}} \right) \nonumber\\
 & \left[1 - \frac{k_{ON} k_A}{k_I (k_{ON} + k_{OFF}) + k_{ON} k_A } \left( 1 + \frac{\alpha_{m_2,I}}{\alpha_{m_2,A}} \frac{\gamma_{m_2} + \alpha_{p_2}}{\gamma_{m_2} + \alpha_{p_1} + \alpha_{p_2}} \frac{k_I}{k_A} \right) \right]\\
 &= \left(\frac{\alpha_{p_2}}{\alpha_{p_3}} \right) \left[ \frac{k_{ON} }{k_I (k_{ON} + k_{OFF}) + k_{ON} k_A }\right] \left( \frac{k_A \alpha_{m_2,A}}{\gamma_{m_2} + \alpha_{p_2}} \right) \nonumber\\
 & \left[1 - \frac{k_{ON} k_A}{k_I (k_{ON} + k_{OFF}) + k_{ON} k_A } \left( 1 + \frac{\env_I^*}{\env_A^*} \right) \right].
\end{align}
Recall that $w = \frac{1}{w_1 w_5}$, so increasing $w_1 w_5$ leads to decreasing $w$. Thus if we want the drug F07\#13 to decrease the amount of $\Pr55$, then we require $\frac{d}{dw} \Pr55^*(w = 1) > 0$. This leads to the condition
\begin{eqnarray}
\frac{k_{ON} k_A}{k_I (k_{ON} + k_{OFF}) + k_{ON} k_A } \left( 1 + \frac{\alpha_{m_2,I}}{\alpha_{m_2,A}} \frac{\gamma_{m_2} + \alpha_{p_2}}{\gamma_{m_2} + \alpha_{p_1} + \alpha_{p_2}} \frac{k_I}{k_A} \right) & < & 1
\end{eqnarray}
If we expand the above inequality, we obtain:
\begin{eqnarray}
\frac{\alpha_{m_2,I}}{\alpha_{m_2,A}}\frac{\gamma_{m_2} + \alpha_{p_2}}{\gamma_{m_2} + \alpha_{p_1} + \alpha_{p_2}} & < & 1 + \frac{k_{OFF}}{k_{ON}},
\end{eqnarray}
or
\begin{eqnarray}
\label{eq:F07_cond}
\frac{\alpha_{m_2,I}}{\alpha_{m_2,A}} & < & \left(1 + \frac{k_{OFF}}{k_{ON}} \right) \left( 1 + \frac{\alpha_{p_1}}{\gamma_{m_2} + \alpha_{p_2}} \right).
\end{eqnarray}
This means that for the drug F07\#13 to be effective in reducing the equilibrium value of $\Pr55$, the ratio between the rates of production $\alpha_{m_2,I}$ and $\alpha_{m_2,A}$ must satisfy the condition in~\eqref{eq:F07_cond}. 
%
Note that this condition is necessary because even though F07\#13 may appear to be effective initially, it may not decrease the equilibrium value of $\Pr55$, see Figure 9.e in \cite{demarino2020differences}.

We remark that this result is only valid close to $w = 1$, so it may not be applicable in general (e.g. $w_{\text{F07\#13}} \approx 0.1$ in T-cell). However, the idea is the same in the general case, so for the drug to be effective, we require $\frac{d}{dw} \Pr55^*(w) > 0$ for $0<w \leq 1$, or equivalently,
\begin{equation}
\left[1 - \frac{k_{ON} k_A}{k_I (k_{ON} + k_{OFF}) + k_{ON} k_A w } \left( w + \frac{\alpha_{m_2,I}}{\alpha_{m_2,A}} \frac{\gamma_{m,2} + \alpha_{p_2}}{\gamma_{m_2} + \alpha_{p_1} + \alpha_{p_2}} \frac{k_I}{k_A} \right) \right] > 0.
\end{equation}
%
%
If this inequality holds, then we can expect the F07\#13 to be effective in decreasing $\Pr55$, which will subsequently decrease production of $\p24$. Similar analysis holds when $\Tat^* < \Tat_{crit}$.
Additionally, if all drugs are considered, e.g. $w_2, w_3, w_4 > 1$, then the condition becomes:

\begin{align}
\label{eq:general_cond_F07}
\left[1 - \frac{k_{ON} k_A w^*}{k_I (k_{ON} + k_{OFF}) + k_{ON} k_A w^* w } \left( \frac{w}{w^*} + \frac{\alpha_{m_2,I}}{\alpha_{m_2,A}} \frac{\gamma_{m,2} + \alpha_{p_2}}{\gamma_{m_2} + \alpha_{p_1} + \alpha_{p_2}} \frac{k_I}{k_A} \right) \right] > 0
\end{align}
where $w^* := w_3 w_4$. This suggests the effect of the F07\#13 drug can be enhanced by the other drugs. Equation~\eqref{eq:general_cond_F07} represents the generalization of~\eqref{eq:F07_cond}, accounting for other drugs and including all values of $w\in(0,1)$.
Figure~\ref{fig:condition_F07_verified} demonstrates an example of the condition~\eqref{eq:F07_cond}.
\begin{figure}
    \centering
    \includegraphics[width=1\linewidth]{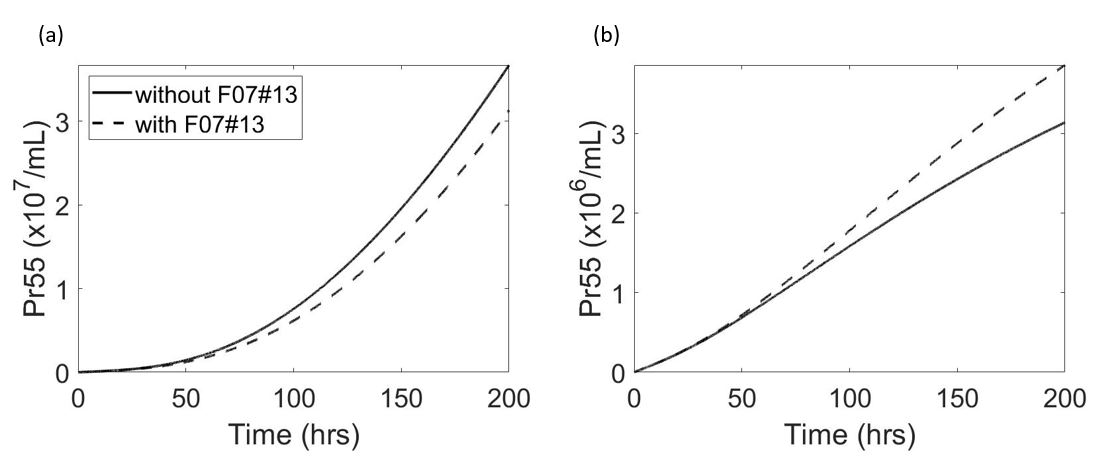}
    \caption{Using the parameters for the linear model for T-cells given in Table~\ref{table:parameters}, we compare the effect of F07\#13 on the dynamics of $\Pr55$ for different levels of $\alpha_{p_1}$.
    (a) Without modification to the parameter, the condition~\eqref{eq:F07_cond} is satisfieda and F07\#13 is effective in reducing the level of $\Pr55$. 
    (b) By reducing the parameter $\alpha_{p_1}$ 1000 folds we break the condition~\eqref{eq:F07_cond}, leading to the ineffectiveness of F07\#13 in reducing the level of $\Pr55$.}
    \label{fig:condition_F07_verified}
\end{figure}

\begin{figure}
    \centering
    \includegraphics[width=1\linewidth]{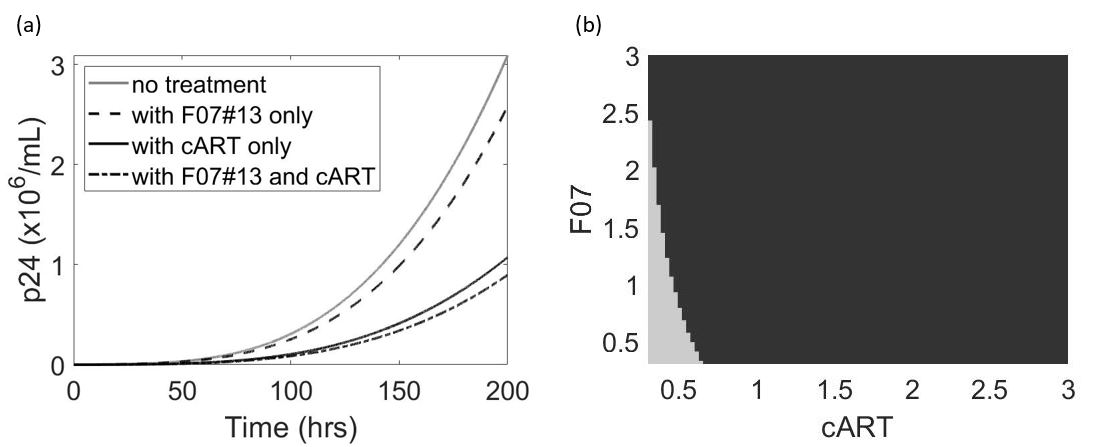}
    \caption{\footnotesize Using the parameters for the linear model for T-cells given in Table~\ref{table:parameters}, we compare the effect of F07\#13 and cART on the dynamics of $p24$. The effect of cART is set at 3 ($w_3 = 3$), meaning it reduces the production of $p24$ from $\Pr55$ to a third. (a) The combination of F07\#13 and cART reduces the level of $p24$ lower than either treatment alone. (b) Level of $p24$ at 200 hours with varying amounts of F07\#13 and cART. Both axes represent the fold change in the drug level of F07\#13 and cART. We use an arbitrary threshold $1.2\times10^6$ to emphasize the effect in varying F07\#13 and cART with respect to one another. The light grey area represents $p24$ level above the threshold, while the dark grey area represents $p24$ level below the threshold.}
    \label{fig:F07_and_cART_synergy}
\end{figure}
The maturation of $\Pr55$ to $\p24$ is targeted by the standard HIV-1 treatment cART (recall that this is represented by the parameter $w_2$ in the model). In Figure~\ref{fig:F07_and_cART_synergy}, we show the possible effects of combining F07\#13 and cART. Figure~\ref{fig:F07_and_cART_synergy}(a) shows that while both cART and F07\#13 are effective in reducing the amount of p24 at 200 hours, when they are used in combination, the level of p24 is reduced further. 
Additionally, we provide a sample synergy map between F07\#13 and cART (without accounting for toxicity) in Figure~\ref{fig:F07_and_cART_synergy}(b). In that figure, an arbitrary $\p24$ level (e.g. $1.2\times10^6$) is used to emphasize the synergistic effect between the two treatments. The light grey area represents $\p24$ level above the threshold, while the dark grey area represents $\p24$ level below the threshold. Since the boundary leans toward a higher dosage of F07\#13, this shows that there is synergy between the two drugs. However, because we do not consider the toxicity level and there is some lack of confidence in the exact values of the drug effect/amount, further study with more comprehensive data is required to estimate the specific value of the synergy between F07\#13 and cART.

\section{Discussion}
\label{sec:disscussion}
%
HIV-1       continues to be a serious problem worldwide. 
Despite tremendous efforts, the ultimate cure for HIV-1 is yet to be discovered. 
Standard treatments, such as cART, target multiple key points in the production of HIV; however, a low level of viral products persists during latency partially due to the lack of an FDA-approved drug to inhibit the viral transcription process. 
This chronic state of HIV-1 is often accompanied by neurocognitive disorders in many patients using cART~\citep{heaton2010hiv,mothobi2012neurocognitive}.
Additionally, experimental drugs often fail during the phase of clinical trials~\citep{khanna2012drug,hwang2016failure}.
This is partially due to a lack of quantitative methods to predict the drug's efficacy and toxicity especially in combination with other drugs.
Thus, a basic understanding of these mechanisms for HIV-1 is crucial for the successful development of new therapies.

In this work, we carry out systematic analyses of the properties of a model of the HIV-1 transcription process that incorporates three distinct promoter states (repressed, intermediate, and activated), introduced in~\cite{demarino2020differences}. To address the discontinuity in the functional response chosen in the original model formulation, we introduce a nonlinear version of the model. The comparison of the two versions of the model reveals interesting biological insights into the transcription process for HIV-1. Finally, a theoretical study of the effectiveness of the experimentally-driven drug F07\#13 is carried out. We summarize and further discuss our findings below.

\textbf{The linear model of ~\citep{demarino2020differences} satisfies basic biological properties.}
For the original system, we show that it is positively invariant given at least one of the LTR states is initially positive.
%
Additionally, we show that all solutions tend to a positive steady state, which is explicitly provided.
The stability of the steady is difficult to link to a biological limit (e.g. resource limitation, etc.). Instead, we should interpret this result as a \textit{stable} increase in the viral HIV-1 load once the activation of transcription starts. 
Furthermore, the availability of the closed form solution allows for direct quantification of the expected viral level and its exponential transcriptional rates, allowing to assess their effect on the transcription process. 

\textbf{Stability results for the extended system.}
To address the discontinuity in the original formulation, we use a continuous functional response that represents the $\Tat$-dependent activation rate. 
We then carry out systematic analyses on the new model. We compare the fitting of the new model to the original version given in \cite{demarino2020differences} to validate its capability to capture the transcriptional dynamics of HIV-1, as shown in  Figures~\ref{fig:env_fitting}-\ref{fig:tar_fitting}.
Furthermore, to show the full capability of the extended model to capture the dynamics of the process, we carry out numerical fitting in Figures~\ref{fig:env_fitting_all_params}-\ref{fig:tar_fitting_all_params}.
The extended system shares many similarities with the original system. For instance, all solutions are positive, bounded and permanent given reasonable initial conditions. It also exhibits a unique positive steady state; however, the stability analysis of this steady state is non-trivial due to the nonlinear functional form of $\Tat$-dependent activation.
By making some simplifying assumptions, we show that the positive steady state is locally asymptotically stable in the special case. 
Furthermore, we see that the extended system can be simplified to the original system in the limit when the production of $\Tat$ is either very high or very low. Thus, we conjecture that the positive steady state is globally stable in the general case.


\textbf{The estimated values of the nonlinear parameters ($v_b, n, m$) suggest distinctive differences between T-cells and macrophages.}
Recall that the larger the values of $n$ and $m$ are, the more alike to a switching function $k_a(\Tat)$ and $f_m(\Tat)$ become, see Figure~\ref{fig:hill_function}.
In T-cells, the large values of $n$ and $m$ (6.23 and 49.8, respectively) suggest more abrupt changes in the level of transcriptional dynamics as $\Tat$ approaches the critical threshold $\Tat_{crit}$, see Table~\ref{table:nonlinear_params}. 
On the other hand, the small value of $n$ and $m$ (both are 1) in macrophages suggest a smoother transition.
Additionally, the fold change $k_A(\Tat)$, $v_b$, is significantly lower than $f_m(\Tat)$, $v_a$, which suggests the increase in promoter activation level before and after $\Tat$ reaches $\Tat_{crit}$ is significantly lower than that of the production of $\env_A$ over the same transition.
Note that these observations still hold even when parameter uncertainty is taken into account, see Table~\ref{table:CI_parameters}.

\textbf{The transcriptional inhibitor F07\#13 is effective in reducing viral production. Furthermore, it is synergetic to standard treatments.}
A useful application of mathematical models is to test the effectiveness of pre-clinical drugs \textit{in silico}. To this end, we study the effect of the HIV-1 transcription inhibitor drug F07\#13. 
The previous simulation in Figure 9(e) in \cite{demarino2020differences} shows that while a drug (e.g. F07\#13) may appear to be ineffective (or effective) initially, the end-result may differ. 
This leads to us establishing a condition that ensures the end-effectiveness of F07\#13 -- that can be applied similarly to other drugs.
Furthermore, we generalize this condition to include the effect of other drugs, which can allow a study of combination therapy to be carried out naturally.
Our simulation and analyses suggest that the incorporation of HIV-1       transcription inhibitors, such as F07, in combination with other HIV-1       treatments may improve their efficacy due to their synergy with one another, see Figure~\ref{fig:F07_and_cART_synergy}.


Using a combination of mathematical analysis and computational simulations, we show interesting observations in the transcriptional dynamics of HIV-1, especially the differentiation in behaviors in the case of T-cells and macrophages. While our model is constructed based on current biological knowledge and validated with experimental data, it is not without limitations.
The model is constructed for the analysis of short term transcriptional dynamics of HIV-1 (on the order of days). Thus, many of the rate parameters are linear, making it unsuitable to study the long term dynamics of HIV-1 (over months or years). 
Additionally, the model does not account for the difference between degradation and exit rates of certain variables.
While the model can be modified to account for the extracellular contents to distinguish between degradation and exit rates, this would further increase the complexity of the model.
Furthermore, we hypothesized the functional forms of the $\Tat$-dependent activation, which may perhaps be improved upon in future attempts.
Finally, while we provide the condition for drug effectiveness in the case of the original model, toxicity may also be included in the case of the nonlinear model to provide stronger insights.
On the analytical side, while we were able to establish basic properties for the linear and nonlinear models in this work, the global stability analysis is still an open question. However, our simulation strongly suggests the existence of a unique globally asymptotically stable positive steady state for the nonlinear model.
With regards to numerical aspects, we carry out basic data fitting and simulations to show differences in dynamical behavior between T-cells and macrophages. 
Perhaps extensive sensitivity analysis can be carried out to aid with the fitting process in the future. Additional data would also allow for better uncertainty quantification for the estimated parameter values and model predictions.
Finally, while we consider the primary effect of the transcriptional inhibitor F07\#13 to reduce the activation rate of LTR, other secondary effects of F07\#13 are not taken into account. Thus, a direct extension would be to account for all known effects of F07\#13 and cART (along with their potential toxicity) in the study of treatment combination. Such a study may prove useful in drug development for clinical application.

\begin{acknowledgements}
We would like to thank all members of the Kashanchi lab, especially Gwen Cox. This work was supported by National Institutes of Health
(NIH) Grants AI078859, AI074410, AI127351-01, AI043894, and NS099029 to F.K., F31NS109443 to C.D., and
George Mason University’s Multidisciplinary Research (MDR) Initiative in Modeling, Simulation and Analytics
funding provided by George Mason University to C.D., D.M.A, M.E., and F.K. Additionally, we
would like to acknowledge all those who participated in Mason Modeling Days 2017 which was funded in part by
the National Science Foundation DMS grant \#1056821 and the College of Science at George Mason University. 
Y.K.  is  partially  supported  by  NSF  grants  DMS-1615879, DEB-1930728 
and  an  NIH  grant5R01GM131405-02.
\end{acknowledgements}

\section*{Availability of data and material}
The data sets and MATLAB code generated during and/or analyzed during the current study are available from the corresponding author on reasonable request.

%
\section*{Conflict of interest}

The authors declare that they have no conflict of interest.

\section{Appendix}
\subsection{Basic properties of the linear model}
\label{appendix:basic_properties}
Note that in the original model, all negative terms in each variable's rate equation is proportional to the variable itself. This implies if the model starts out with positive initial conditions, it will never become negative. 
The subsystem $\LTR_R, \LTR_I, \LTR_A$ is conserved, so the zero steady state is always unstable given at least one of the $\LTR$ state is positive initially.
Furthermore, since the system (equation~\ref{eq:main}-\ref{eq:main_end}) has a globally stable unique positive steady state (see Theorem~\ref{theorem:gas_ori}), the system is bounded eventually.
Thus we arrive at the proposition.
\begin{proposition}
The system in Equation~\ref{eq:main}-\ref{eq:main_end} is positively invariant and eventually bounded.
\end{proposition}

\subsection{Calculation of steady states for the case of constant rate between the intermediate to activated LTR and piecewise production of activated envelope}
\label{appendix:steady_state_linear}

In dimension-reduced form (using $1 = \LTR_R + \LTR_I + \LTR_A$), the rate of change of $\LTR_I$ becomes:

\begin{eqnarray}
\frac{d}{dt} \left[ \LTR_I \right] & = & k_{ON} - \left(\frac{w_3 w_4}{w_5}k_A + k_{OFF} + k_{ON}\right) \LTR_I \nonumber \\
& &+ (w_1 k_I - k_{ON}) \LTR_A
\end{eqnarray}

Setting $\LTR_A' = 0$ and solve for $\LTR_I^*$ in term of $\LTR_A^*$ gives:

\begin{eqnarray}
\LTR_I^* & = & \frac{w_1 w_5}{w_3 w_4} \frac{k_I}{k_A} \LTR_A^*
\end{eqnarray}

Replacing $\LTR_I^*$ into $\LTR_I' = 0$ to solve for $\LTR_A^*$:

\begin{eqnarray}
\LTR_A^* & = & \frac{w_3 w_4}{w_1 w_5} \frac{k_{ON} k_A}{k_I (k_{ON} + k_{OFF}) + k_{ON} k_A \left(\frac{w_3 w_4}{w_1 w_5}\right)}
\end{eqnarray}

Using the relation between $\LTR_I^*$ and $\LTR_A^*$, we obtain:

\begin{eqnarray}
\LTR_I^* & = & \frac{w_1 w_5}{w_3 w_4} \frac{k_I}{k_A} \LTR_A^*\\
& = & \frac{k_{ON} k_I}{k_I (k_{ON} + k_{OFF}) + k_{ON} k_A \left(\frac{w_3 w_4}{w_1 w_5}\right)}
\end{eqnarray}

And similarly for $\LTR_R^*$:

\begin{eqnarray}
\LTR_R^* & = & \frac{k_{OFF}}{k_{ON}} \LTR_I^* \\
 & = & \frac{k_{OFF} k_I}{k_I (k_{ON} + k_{OFF}) + k_{ON} K_A \left(\frac{w_3 w_4}{w_1 w_5}\right)}
\end{eqnarray}

We proceed to compute the steady states of the remaining variables. In the asymptotic limit, then $\env_I'$ becomes

\begin{eqnarray}
\env_I' & = & \alpha_{m_2,I} \LTR_I^* - (\gamma_{m,2} + \alpha_{p_1} + \alpha_{p,2})\env_I
\end{eqnarray}

It follows that

\begin{eqnarray}
\env_I^* & = & \frac{\alpha_{m_2,I} \LTR_I^*}{\gamma_{m,2} + \alpha_{p_1} + \alpha_{p_2}}
\end{eqnarray}

Similarly, 

\begin{eqnarray}
\Tat^* & = & \frac{\alpha_{p_1} \env_I^*}{\gamma_{p_1}}
\end{eqnarray}

And

\begin{eqnarray}
\TAR^* & = & \frac{\alpha_{m_1,R}\LTR_A^* + \alpha_{m_1,I}\LTR_I^* + \alpha_{m_1,A}\LTR_A}{\gamma_{m_1}}
\end{eqnarray}

Then,

\begin{eqnarray}
\env_A^* = 
\left\{ 
\begin{array}{ll}
\frac{\alpha_{m_2,A} \LTR_A^*}{v_a (\gamma_{m_2} + \alpha_{p_2}) } & \mbox{if} \; \Tat^* < \Tat_{crit} \\
\frac{\alpha_{m_2,A} \LTR_A^*}{\gamma_{m_2} + \alpha_{p_2} } &  \mbox{if} \;  \Tat^* \ge \Tat_{crit}
\end{array} 
\right.
\end{eqnarray}

where $\Tat^* = \frac{\alpha_{p_1} \env_I^*}{\gamma_{p_1}}$. Finally,

\begin{eqnarray}
Pr55^* & = & \frac{\alpha_{p_2} \env_I^* + \alpha_{p_2} \env_A^*}{\alpha_{p_3}}
\end{eqnarray}

and

\begin{eqnarray}
p24^* & = & \frac{\alpha_{p_3}}{\gamma_{p_2}} \frac{Pr55^*}{w_2}
\end{eqnarray}

\subsection{Confidence interval on parameter estimation}
The function fmincon does not provide the Jacobian matrix needed to calculate the confidence interval for parameter estimation. Instead, we utilize the function lsqnonlin, a different MATLAB function for nonlinear fitting, to establish confidence on our parameter estimates. Due to the biologically realistic constraints on some parameters in the optimization process, the confidence intervals established here may not be reliable for certain parameters. However, it still provides a good estimate of the uncertainty related to the model and the available data. We note that the fitting results are comparable between lsqnonlin and fmincon; however, we elect to use fmincon in the main draft to be consistent with the original publication~\citep{demarino2020differences}.

\begin{table}
\begin{center}
\begin{tabular}{lcc}
\hline
 & T-Cell    & macrophages \\ \hline\hline
$\alpha_{m_1,A}$ & $1.70 \times 10^7 [0, 4.4\times 10^7]$   & $4.56 \times 10^5 [2.52\times 10^5,6.57\times 10^5]$ \\[.75 ex]
$\alpha_{m_2,A}$ & $2.47 \times 10^6 [0, 1.09\times 10^7]$   & $4.04 \times 10^4 [0,1.32\times 10^5]$ \\[.75 ex]
$\gamma_{m,1}$   & $2.63 \times 10^4 [0,8.41\times 10^4]$    & $2.68 \times 10^4 [1.16\times 10^4,4.19\times 10^6]$ \\[.75 ex]
$\gamma_{m,2}$   & $2.00 \times 10^3 [0,8.08\times 10^4]$    & $5.91 \times 10^2 [0,5.31\times 10^4]$ \\[.75 ex]                 
\hline\hline
$\alpha_{m_1,A}$ & $1.30 \times 10^7 [0,1.24\times 10^8]$   & $5.28 \times 10^5 [0,2.12\times 10^6]$ \\[.75 ex]
$\alpha_{m_2,A}$ & $1.75 \times 10^6 [0,2.03\times 10^7]$   & $5.32 \times 10^4 [0,5.03\times 10^5]$ \\[.75 ex]
$\gamma_{m,1}$   & $1.82 \times 10^4 [0,2.41\times 10^5] $    & $2.68 \times 10^4[0,7.00\times 10^4]$ \\[.75 ex]
$\gamma_{m,2}$   & $1.10 \times 10^3 [0,1.88\times 10^5]$    & $5.91 [0,1.58\times 10^5]$ \\[.75 ex]
$v_b$          & $1.47\times 10^2 [10,1.31\times 10^7]$  & $10 [10,5.71\times 10^2]$ \\ 
$n$            & $2.10 [1,2.05\times 10^3]$  & $1 [1,1.65\times 10^1]$ \\ [.75 ex]
$m$            & $1.34 \times 10^2[1,1.19\times 10^8]$  & $1 [1,2.28\times 10^1]$ \\ [.75 ex]
\hline
\end{tabular} \\
\end{center}
\vspace{0.0in}
\caption{Confidence intervals of estimated parameters using lsqnonlin for the linear (top half) and non-linear (bottom half) models. The estimated value is provided with the lower and upper 95\% confidence interval. Note that when the lower bound of the 95\% confidence interval is not biologically reasonable, we set it to the biological bound. The estimated values of the parameters are similar to that of fmincon.}
\label{table:CI_parameters}
\end{table} 

\subsection{Fitting using all parameters}
In Figures~\ref{fig:env_fitting} and~\ref{fig:tar_fitting}, we only attempt to fit four (linear model) and seven (nonlinear model) parameters. This reservation is due to the limited available data and to avoid potential problem of identifiability. However, in this subsection, we will demonstrate the flexibility of both models to capture the dynamics of the HIV-1 transcription process by fitting to all parameters (with the exception of the drug-related parameters $w_i$).
The results are demonstrated in Figures~\ref{fig:env_fitting_all_params} and~\ref{fig:tar_fitting_all_params}.

\begin{figure}[!ht]
\centering
\includegraphics[width=1\linewidth]{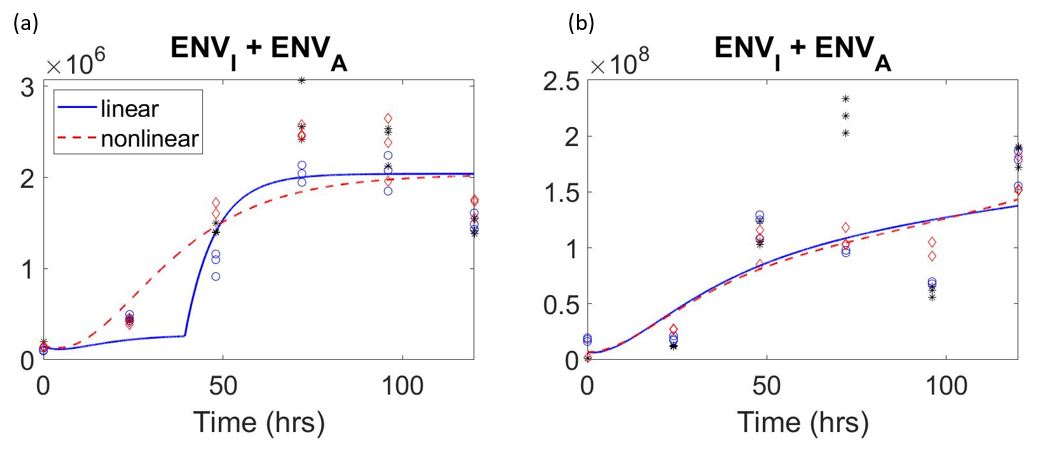}
\caption{\footnotesize Comparison of linear and nonlinear model for env-data fitting in both T-cells and macrophages. (a) ENV data and fitting in macrophages. (b) ENV data and fitting in T-cells. The fitting is carried out using all parameters.}
        \label{fig:env_fitting_all_params}
\end{figure}
\begin{figure}[!ht]
\centering
\includegraphics[width=1\linewidth]{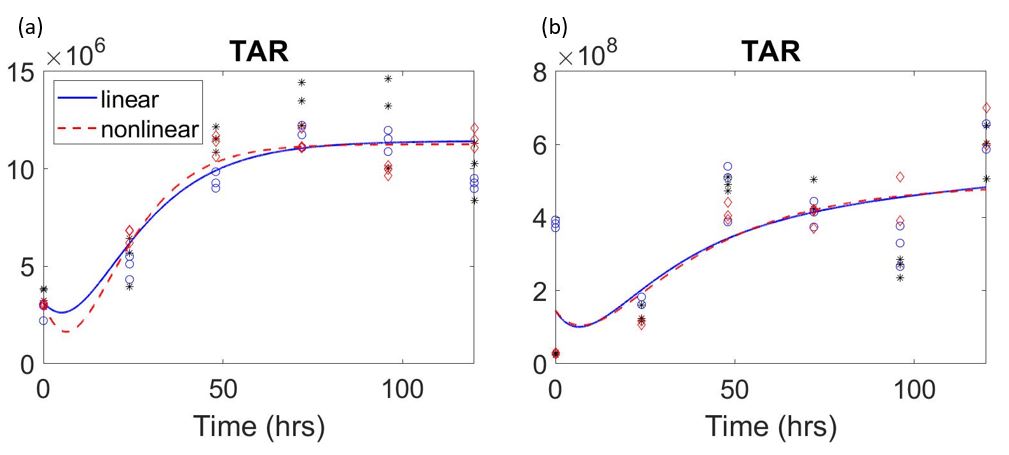}
\caption{\footnotesize Comparison of linear and nonlinear model for TAR-data fitting in both T-cells and macrophages. (a) TAR data and fitting in macrophages. (b) TAR data and fitting in T-cells. The fitting is carried out using all parameters.}
        \label{fig:tar_fitting_all_params}
\end{figure}

\subsection{Exact solution in the case of constant rates}
\label{appendix:exact_sol_derived}

The equations examined here correspond to the model defined in the main text with all coefficients taken to be constant.  To simplify notation we have not included the drug treatment factors $w_i$ but those can be easily incorporated by a re-interpretation of the coefficients.

\begin{eqnarray}
 \frac{d}{dt} \left[ \LTR_R \right]  &  =  & - k_{ON} \LTR_R + k_{OFF} \LTR_I  \\
 \frac{d}{dt} \left[ \LTR_I \right]   & = & k_{ON} \LTR_R - \left[k_A + k_{OFF} \right] \LTR_I + k_I \LTR_A   \\
\frac{d}{dt} \left[ \LTR_A \right]  & = & k_A \;\LTR_I -  k_I  \LTR_A \\
 \frac{d}{dt} \left[ \TAR \right]  & = &    \alpha_{m_1,R} \LTR_R  + \alpha_{m_1,I} \LTR_I + \alpha_{m_1,A} \LTR_A - \gamma_{m_1} \TAR  \\
 \frac{d}{dt} \left[ \env_{I} \right]  & = &   \alpha_{m_2,I} \LTR_I - \left( \gamma_{m,2} + \alpha_{p_1} + \alpha_{p_2} \right) \env_{I}  \\ 
 \frac{d}{dt} \left[ \env_{A} \right]  & = &   \alpha_{m_2,A} \LTR_A - \left( \gamma_{m,2} + \alpha_{p_2} \right) \env_{A}   \\
\frac{d}{dt} \left[ \Tat \right]  & = &  \alpha_{p_1} \env_{I} - \gamma_{p_1}  \Tat  \\
\frac{dP_{r55}}{dt} & = & \alpha_{p_2} \env_{I} + \alpha_{p_2} \env_{A} - \alpha_{p_3} P_{r55}  \\
\frac{dP_{24}}{dt} & = & \alpha_{p_3} P_{r55}  - \gamma_{p_2} P_{24} 
\end{eqnarray}

Note that in the $\Tat$ equation we have used a constant coefficient $\alpha_{m_2,A}$ as the multiplier on the $\LTR_A$ term.  In our full model we also examined the piecewise-constant case where
$\alpha_{m_2,A}/v_a$ if $\Tat < \Tat_{crit}$ and $\alpha_{m_2,A}$ if $\Tat \ge \Tat_{crit}$.  The solutions presented below can be adapted to address this piecewise-constant
situation by reusing the formulas starting at a new time $t=t^0$ and with new initial conditions $\LTR_R(t=t^0) := \LTR_R^0$, $\LTR(t=t^0) := \LTR_I^0$, etc. along with a numerical calculation of $t^0=t_{crit}$ as the time
for which $\Tat$ reaches $\Tat_{crit}$.  We have validated this approach by comparing the exact solution with piecewise-constants with numerically-computed solutions using MATLAB's ode solver {\tt ode23s}.  

\textbf{LTR equations}
Suppose $k_{OFF}$, $k_{ON}$, $k_A$ and $k_I$ are constant.  The $\LTR$ equations decouple from the rest and can be solved first.
\begin{eqnarray}
 \frac{d}{dt} \left[ \LTR_R \right]  &  =  & - k_{ON} \LTR_R + k_{OFF} \LTR_I  \\
 \frac{d}{dt} \left[ \LTR_I \right]   & = & k_{ON} \LTR_R - \left[k_A + k_{OFF} \right] \LTR_I + k_I \LTR_A   \\
\frac{d}{dt} \left[ \LTR_A \right]  & = & k_A \;\LTR_I -  k_I  \LTR_A.
\end{eqnarray}
This is a solvable linear system that can be expressed in matrix form
\begin{eqnarray}
\label{eq:original_linear_system}
\frac{d\vec{L}}{dt} & = & 
\left[
\begin{array}{ccc}
-a & b & 0 \\
a & - b - c & d \\
0 & c & - d
\end{array}
\right] \vec{L} \equiv A \;\vec{L},
\end{eqnarray}
where 
\begin{eqnarray}
\vec{L} & = &
\left[
\begin{array}{c}
\LTR_R \\
\LTR_I \\
\LTR_A
\end{array}
\right],
\end{eqnarray}
and for convenience we introduce matrix entries
\begin{eqnarray}
a = k_{ON},\quad
b = k_{OFF},\quad
c = k_A, \quad
d = k_I.
\end{eqnarray}

Solutions to equation~(\ref{eq:original_linear_system}) have the form $\vec{v} e^{\lambda t}$ with eigenvector $\vec{v}$ and eigenvalue $\lambda$ satisfying $A \vec{v} = \lambda \vec{v}$.
In particular, the eigenvalues of $A$ satisfy
\begin{eqnarray}
0 & = & \lambda \left[ \lambda^2 + (a + b + c + d) \lambda + (ac + ad + bd)  \right].
\end{eqnarray}
So $\lambda=\lambda_0= 0$, $\lambda=\lambda_1$ and $\lambda= \lambda_2$ where
\begin{eqnarray}
\lambda_{1,2} & = & \frac{1}{2} \left[ - (a + b + c + d) \pm \sqrt{ (a + b + c + d)^2 - 4 (ac + ad + bd)  } \right]
\label{lambda}
\end{eqnarray}
The eigenvectors corresponding to $\lambda_0$, $\lambda_1$, and $\lambda_2$ are
\begin{eqnarray}
\vec{v}_i & = & \left[
\begin{array}{c}
b (d+\lambda_i) \\
(a+\lambda_i) (d+\lambda_i) \\
c (a+ \lambda_i)
\end{array}
\right] \hspace{0.25in} \mbox{for $i=0,1,2$.}
\end{eqnarray}
Then,
\begin{eqnarray}
\label{eq:LTR_solution}
\vec{L} = \left[
\begin{array}{c}
\LTR_R(t) \\
\LTR_I(t) \\
\LTR_A(t)
\end{array}
\right] & = & \sum_{i=0}^{2} c_i \vec{v}_i e^{\lambda_i (t-t^0)},
\end{eqnarray}
where the constants $c_i$ for $i=0,1,2$ are solutions of the linear system
\begin{eqnarray}
\label{eq:c_linear_system} 
\left[
\begin{array}{ccc}
bd & b(d+\lambda_1) & b(d+\lambda_2) \\
ad & (a+\lambda_1)(d+\lambda_1) & (a+\lambda_2)(d+\lambda_2) \\
ca & c(a+\lambda_1) & c(a +\lambda_2)
\end{array}
\right] 
\left[
\begin{array}{c}
c_0 \\
c_1 \\
c_2
\end{array}
\right]
=
\left[
\begin{array}{c}
\LTR_R^0 \\
\LTR_I^0 \\
\LTR_A^0
\end{array}
\right],
\end{eqnarray}
and $\LTR_R^0$, $\LTR_I^0$, and $\LTR_A^0$ are the three $\LTR$ values at $t=t^0$ (Note that $t^0$ could be zero or some other suitable time).  Once the values of $c_0$, $c_1$, and $c_2$ are determined then equation~(\ref{eq:LTR_solution}) represents the $\LTR$ solution for the case of constant $k_{ON}$, $k_{OFF}$, $k_A$, and $k_I$.

\textbf{RNA equations - $\TAR$, $\env_I$, $\env_A$}

The equations for the short and long RNA sequences are
\begin{eqnarray}
 \frac{d}{dt} \left[ \TAR \right]  & = &    \alpha_{m_1,R} \LTR_R  + \alpha_{m_1,I} \LTR_I + \alpha_{m_1,A} \LTR_A - \gamma_{m_1} \TAR  \\
 \frac{d}{dt} \left[ \env_{I} \right]  & = &   \alpha_{m_2,I} \LTR_I - \left( \gamma_{m,2} + \alpha_{p_1} + \alpha_{p_2} \right) \env_{I}  \\ 
 \frac{d}{dt} \left[ \env_{A} \right]  & = &   \alpha_{m_2,A} \LTR_A - \left( \gamma_{m,2} + \alpha_{p_2} \right) \env_{A} 
\end{eqnarray}

Note that since $\LTR_R$, $\LTR_I$, and $\LTR_A$ are linear combinations of exponential functions (or a constant) all three of these equations have the general form
\bea
\frac{dy}{dt} + \gamma y & = & p_0 + \sum_{i=1}^{2} p_i e^{\lambda_i (t-t^0)}.
\eea
Introducing an integrating factor $e^{\gamma (t-t^0)}$, integrating, and using $y(t^0)=y^0$ leads to the solution
\bea
 y & = & y^0 e^{-\gamma (t-t^0)} + \frac{p_0}{\gamma} \left( 1 - e^{-\gamma (t-t^0)} \right) \nonumber \\
 & & + \sum_{i=1}^{2} \frac{p_i}{\lambda_i+\gamma} \left( e^{\lambda_i (t-t^0)} - e^{-\gamma (t - t^0)}\right),\hspace{0.10in}\mbox{if $\gamma \neq 0$}, \\
  y & = & y^0 + p_0 (t-t^0) + \sum_{i=1}^{2} \frac{p_i}{\lambda_i} \left( e^{\lambda_i (t-t^0)} - 1\right),\hspace{0.25in}\mbox{if $\gamma = 0$}.
\eea

In the following three subsections we write out the values for $\gamma$, $p_0$, $p_1$, and $p_2$ for the $\TAR$, $\env_I$, and $\env_A$ equations.

\textbf{$\TAR$ equation}

\bea
 \frac{d}{dt} \left[ \TAR \right] + \gamma_{m_1} \TAR & = &  \alpha_{m_1,R} \LTR_R  + \alpha_{m_1,I} \LTR_I + \alpha_{m_1,A} \LTR_A   \nonumber \\
& = & 
\alpha_{m_1,R}  \left( c_0 bd + c_1 b(d+\lambda_1) e^{\lambda_1 (t-t^0)} + c_2 b(d+\lambda_2) e^{\lambda_2 (t-t^0)} \right) \nonumber \\
& & \mbox{} + \alpha_{m_1,I} \left( c_0 ad + c_1 (a+\lambda_1)(d+\lambda_1)e^{\lambda_1 (t-t^0)} + c_2 (a+\lambda_2)(d+\lambda_2) e^{\lambda_2 (t-t^0)} \right), \nonumber  \\
& & \mbox{} + \alpha_{m_1,A} \left( c_0 ca + c_1 c(a+\lambda_1) e^{\lambda_1 (t-t^0)} + c_2 c(a+\lambda_2) e^{\lambda_2 (t-t^0)}\right)
 \eea
 so for this case
 \bea
 \gamma & = & \gamma_{m_1} \equiv \gamma^{(TAR)}, \\
 p_0 & = & c_0 \left( \alpha_{m_1,R} bd + \alpha_{m_1,I} ad+  \alpha_{m_1,A} ac \right) \equiv p_0^{(TAR)},\\
 p_1 & = & c_1 \left( \alpha_{m_1,R} b(d+\lambda_1) + \alpha_{m_1,I}(a+\lambda_1)(d+\lambda_1) +  \alpha_{m_1,A} c(a+\lambda_1) \right) \nonumber\\
 & & \equiv p_1^{(TAR)} ,\\
 p_2 & = & c_2 \left( \alpha_{m_1,R} b(d+\lambda_2) + \alpha_{m_1,I} (a+\lambda_2)(d+\lambda_2)+  \alpha_{m_1,A} c(a+\lambda_2) \right) \nonumber \\
 & & \equiv p_2^{(TAR)}.
 \eea
Then,
\bea
 \TAR & = & \TAR^0 e^{-\gamma (t-t^0)} + \frac{p_0}{\gamma} \left( 1 - e^{-\gamma (t-t^0)} \right) \nonumber \\
 & & \mbox{} + \sum_{i=1}^{2} \frac{p_i}{\lambda_i+\gamma} \left( e^{\lambda_i (t-t^0)} - e^{-\gamma (t - t^0)}\right),\hspace{0.25in}\mbox{if $\gamma \neq 0$}, \\
  \TAR & = & \TAR^0 + p_0 (t-t^0) + \sum_{i=1}^{2} \frac{p_i}{\lambda_i} \left( e^{\lambda_i (t-t^0)} - 1\right),\hspace{0.25in}\mbox{if $\gamma = 0$}.
\eea

\textbf{$\env_I$ equation}

\bea
 \frac{d}{dt} \left[ \env_{I} \right]  +  \left( \gamma_{m,2} + \alpha_{p_1} + \alpha_{p_2} \right) \env_{I} & = &   \alpha_{m_2,I} \LTR_I  \nonumber  \\ 
  & = & \alpha_{m_2,I} \left( c_0 ad + c_1 (a+ \lambda_1)(d+\lambda_1) e^{\lambda_1 (t-t^0)} \right. \nonumber \\
  & & \mbox{} \left. + c_2 (a+ \lambda_2)(d+\lambda_2) e^{\lambda_2 (t-t^0)} \right), 
 \eea
 so for this case
 \bea
 \gamma & = & \gamma_{m,2} + \alpha_{p_1} + \alpha_{p_2} \equiv \gamma^{(I)}, \\
 p_0 & = & \alpha_{m_2,I} c_0 ad \equiv p_0^{(I)},\\
 p_1 & = & \alpha_{m_2,I} c_1 (a+ \lambda_1)(d+\lambda_1) \equiv p_1^{(I)},\\
 p_2 & = & \alpha_{m_2,I} c_2 (a+ \lambda_2)(d+\lambda_2) \equiv p_2^{(I)}.
 \eea
 Then,
 \bea
 \env_I & = & \env_I^0 e^{-\gamma (t-t^0)} + \frac{p_0}{\gamma} \left( 1 - e^{-\gamma (t-t^0)} \right) \nonumber \\
 & & \mbox{} + \sum_{i=1}^{2} \frac{p_i}{\lambda_i+\gamma} \left( e^{\lambda_i (t-t^0)} - e^{-\gamma (t - t^0)}\right),\hspace{0.25in}\mbox{if $\gamma \neq 0$}, \\
  \env_I & = & \env_I^0 + p_0 (t-t^0) + \sum_{i=1}^{2} \frac{p_i}{\lambda_i} \left( e^{\lambda_i (t-t^0)} - 1\right),\hspace{0.25in}\mbox{if $\gamma = 0$}.
\eea

\textbf{$\env_A$ equation}

\bea
 \frac{d}{dt} \left[ \env_{A} \right]  +  \left( \gamma_{m,2} + \alpha_{p_2} \right) \env_{A}& = &   \alpha_{m_2,A} \LTR_A   \\
  & = & \alpha_{m_2,A} \left( c_0 ca + c_1 c (a+\lambda_1) e^{\lambda_1 (t-t^0)} + c_2 c(a+\lambda_2) e^{\lambda_2 (t-t^0)} \right), \nonumber 
 \eea
so for this case
 \bea
 \gamma & = & \gamma_{m,2} + \alpha_{p_2} \equiv \gamma^{(A)}, \\
 p_0 & = & \alpha_{m_2,A} c_0 ca  \equiv p_0^{(A)},\\
 p_1 & = & \alpha_{m_2,A} c_1 c(a+\lambda_1)\equiv p_1^{(A)},\\
 p_2 & = & \alpha_{m_2,A} c_2 c(a+\lambda_2)\equiv p_2^{(A)}.
 \eea
Then,
 \bea
 \env_A & = & \env_A^0 e^{-\gamma (t-t^0)} + \frac{p_0}{\gamma} \left( 1 - e^{-\gamma (t-t^0)} \right) \nonumber \\
 & & \mbox{} + \sum_{i=1}^{2} \frac{p_i}{\lambda_i+\gamma} \left( e^{\lambda_i (t-t^0)} - e^{-\gamma (t - t^0)}\right),\hspace{0.25in}\mbox{if $\gamma \neq 0$}, \\
  \env_A & = & \env_A^0 + p_0 (t-t^0) + \sum_{i=1}^{2} \frac{p_i}{\lambda_i} \left( e^{\lambda_i (t-t^0)} - 1\right),\hspace{0.25in}\mbox{if $\gamma = 0$}.
\eea

\textbf{$\Tat$ equation}

\bea
\frac{d}{dt} \left[ \Tat \right] +\gamma_{p_1} \Tat & = &  \alpha_{p_1} \env_{I}.
\eea

In this subsection for $\Tat$ we have
 \bea
 y^0 & = & \env_I^0,\\
 \gamma & = & \gamma_{m,2} + \alpha_{p_1} + \alpha_{p_2} = \gamma^{(I)}, \\
 p_0 & = & \alpha_{m_2,I} c_0 ad = p_0^{(I)},\\
 p_1 & = & \alpha_{m_2,I} c_1 (a+ \lambda_1)(d+\lambda_1) = p_1^{(I)},\\
 p_2 & = & \alpha_{m_2,I} c_2 (a+ \lambda_2)(d+\lambda_2) = p_2^{(I)}.
 \eea

\textbf{case} $\gamma \neq 0$, $\gamma_{p_1} \neq 0$, $\gamma \neq \gamma_{p_1}$

Based on the form of $\env_I$, and noting $\gamma = \gamma_{m,2} + \alpha_{p_1} + \alpha_{p_2}$, this equation has the form
\bea
\frac{d \Tat}{dt} + \gamma_{p_1} \Tat & = & \alpha_{p_1} \left[
 y^0 e^{-\gamma (t-t^0)} + \frac{p_0}{\gamma} \left( 1 - e^{-\gamma (t-t^0)} \right) + \sum_{i=1}^{2} \frac{p_i}{\lambda_i+\gamma} \left( e^{\lambda_i (t-t^0)} - e^{-\gamma (t - t^0)}\right) \right],\nonumber \\
  & = & \alpha_{p_1} \frac{p_0}{\gamma} +  \alpha_{p_1} \left[ \sum_{i=1}^{2} \frac{p_i}{\lambda_i+\gamma}  e^{\lambda_i (t-t^0)}  \right] + \alpha_{p_1} \left[ y^0 - \frac{p_0}{\gamma} - \sum_{i=1}^{2} \frac{p_i}{\lambda_i+\gamma} \right] e^{-\gamma (t-t^0)},\nonumber \\
  & = & q_0 + \sum_{i=1}^{2} q_i e^{\lambda_i (t-t^0)} + q_3 e^{-\gamma (t-t^0)}.
  \eea
where
\bea
q_0 & = & \alpha_{p_1} \frac{p_0}{\gamma} \nonumber \\
q_i & = & \alpha_{p_1} \frac{p_i}{\lambda_i+\gamma} \hspace{0.25in} \mbox{for $i=1,2$} \nonumber \\
q_3 & = & \alpha_{p_1} \left[ y^0 - \frac{p_0}{\gamma} - \sum_{i=1}^{2} \frac{p_i}{\lambda_i+\gamma} \right].
\eea
 In this case
 \bea
 \Tat & = & \Tat^0 e^{-\gamma_{p_1} (t-t^0)} + \frac{q_0}{\gamma_{p_1}} \left[ 1 - e^{-\gamma_{p_1} (t-t^0)} \right] \nonumber \\
  & & \mbox{} + \sum_{i=1}^{2} \frac{q_i}{\lambda_i + \gamma_{p_1}} \left[ e^{\lambda_i (t-t^0)} - e^{-\gamma_{p_1} (t-t^0)} \right] \nonumber \\
  & & \mbox{} + \frac{q_3}{\gamma_{p_1} - \gamma} \left[ e^{-\gamma (t-t^0)} - e^{-\gamma_{p_1} (t-t^0)} \right]
 \eea

\textbf{case} $\gamma \neq 0$, $\gamma_{p_1} = 0$

Here the equation for $\Tat$ is 
\bea
\frac{d \Tat}{dt}  & = & q_0 + \sum_{i=1}^{2} q_i e^{\lambda_i (t-t^0)} + q_3 e^{-\gamma (t-t^0)},
\eea
where
\bea
q_0 & = & \alpha_{p_1} \frac{p_0}{\gamma} \nonumber \\
q_i & = & \alpha_{p_1} \frac{p_i}{\lambda_i+\gamma} \hspace{0.25in} \mbox{for $i=1,2$} \nonumber \\
q_3 & = & \alpha_{p_1} \left[ y^0 - \frac{p_0}{\gamma} - \sum_{i=1}^{2} \frac{p_i}{\lambda_i+\gamma} \right].\
\eea
 In this case
 \bea
 \Tat & = & \Tat^0 + q_0 (t-t^0) + \sum_{i=1}^{2} \frac{q_i}{\lambda_i} \left[ e^{\lambda_i (t-t^0)} - 1 \right]  - \frac{q_3}{\gamma} \left[ e^{-\gamma (t-t^0)} - 1 \right]
 \eea

\textbf{case} $\gamma = \gamma_{p_1} \neq 0$

Here the equation for $\Tat$ is
\bea
\frac{d \Tat}{dt} + \gamma \Tat 
  & = & q_0 + \sum_{i=1}^{2} q_i e^{\lambda_i (t-t^0)} + q_3 e^{-\gamma (t-t^0)}.
  \eea
where
\bea
q_0 & = & \alpha_{p_1} \frac{p_0}{\gamma} \nonumber \\
q_i & = & \alpha_{p_1} \frac{p_i}{\lambda_i+\gamma} \hspace{0.25in} \mbox{for $i=1,2$} \nonumber \\
q_3 & = & \alpha_{p_1} \left[ y^0 - \frac{p_0}{\gamma} - \sum_{i=1}^{2} \frac{p_i}{\lambda_i+\gamma} \right].
\eea
 In this case
 \bea
 \Tat & = & \Tat^0 e^{-\gamma (t-t^0)} + \frac{q_0}{\gamma} \left[ 1 - e^{-\gamma (t-t^0)} \right] \nonumber \\
  & & \mbox{} + \sum_{i=1}^{2} \frac{q_i}{\lambda_i + \gamma} \left[ e^{\lambda_i (t-t^0)} - e^{-\gamma (t-t^0)} \right] 
  + q_3 (t-t^0) e^{-\gamma (t-t^0)} 
 \eea

\textbf{case} $\gamma = 0$, $\gamma_{p_1} \neq 0$

Here the equation for $\Tat$ is
\bea
\frac{d \Tat}{dt} + \gamma_{p_1} \Tat & = & q_0 + \sum_{i=1}^{2} q_i e^{\lambda_i (t-t^0)} + q_3 (t-t^0).
  \eea
where
\bea
q_0 & = & \alpha_{p_1} \left[ y^0 - \sum_{i=1}^{2} \frac{p_i}{\lambda_i} \right]  \nonumber \\
q_i & = & \alpha_{p_1} \frac{p_i}{\lambda_i} \hspace{0.25in} \mbox{for $i=1,2$} \nonumber \\
q_3 & = & \alpha_{p_1} p_0.
\eea
In this case
 \bea
 \Tat & = & \Tat^0 e^{-\gamma_{p_1} (t-t^0)} + \frac{q_0}{\gamma_{p_1}} \left[ 1 - e^{-\gamma_{p_1} (t-t^0)} \right] \nonumber \\
  & & \mbox{} + \sum_{i=1}^{2} \frac{q_i}{\lambda_i + \gamma_{p_1}} \left[ e^{\lambda_i (t-t^0)} - e^{-\gamma_{p_1} (t-t^0)} \right] \nonumber \\
  & & \mbox{} + \frac{q_3}{\gamma_{p_1}^2} \left[ e^{-\gamma_{p_1} (t-t^0)}  - 1 + \gamma_{p_1} (t-t^0)  \right]
 \eea

\textbf{case} $\gamma = \gamma_{p_1} = 0$

Here the equation for $\Tat$ is
\bea
\frac{d \Tat}{dt} & = & q_0 + \sum_{i=1}^{2} q_i e^{\lambda_i (t-t^0)} + q_3 (t-t^0).
  \eea
where
\bea
q_0 & = & \alpha_{p_1} \left[ y^0 - \sum_{i=1}^{2} \frac{p_i}{\lambda_i} \right]  \nonumber \\
q_i & = & \alpha_{p_1} \frac{p_i}{\lambda_i} \hspace{0.25in} \mbox{for $i=1,2$} \nonumber \\
q_3 & = & \alpha_{p_1} p_0.
\eea
In this case
\bea
 \Tat & = & \Tat^0 + q_0 (t-t^0) + \sum_{i=1}^{2} \frac{q_i}{\lambda_i} \left[ e^{\lambda_i (t-t^0)} - 1 \right] 
  + \frac{1}{2} q_3 (t-t^0)^2 
 \eea

\textbf{$P_{r55}$ equation}

\bea
\frac{dP_{r55}}{dt} + \alpha_{p_3} P_{r55} & = & \alpha_{p_2} \env_{I} + \alpha_{p_2} \env_{A}  
\eea

In this case $\alpha_{p_3} \neq 0$ (we do not consider $\alpha_{p_3}=0$).  Also we'll assume $\gamma^{(I)} \neq 0$ and $\gamma^{(A)} \neq 0$.

Inserting the results for $\env_I$ and $\env_A$ and also introducing notation $\alpha_{p_2}^{(I)}$ and $\alpha_{p_2}^{(A)}$ for the two $\alpha_{p_2}$ coefficients 
leads to
\bea
\frac{d}{dt} \left[ e^{\alpha_{p_3} (t-t^0)}P_{r55} \right] & = & \alpha_{p_2}^{(I)} \env_{I}^0 e^{(\alpha_{p_3}-\gamma^{(I)})(t-t^0)} + 
\alpha_{p_2}^{(I)} \frac{p_0^{(I)} }{\gamma^{(I)} } \left( e^{\alpha_{p_3}(t-t^0)} - e^{(\alpha_{p_3} - \gamma^{(I)}) (t-t^0)} \right) \nonumber \\
& & \mbox{} + \alpha_{p_2}^{(I)} \sum_{i=1}^{2} \frac{p_i^{(I)} }{\lambda_i + \gamma^{(I)} } \left(   e^{(\alpha_{p_3}+\lambda_i)(t-t^0)}  - e^{(\alpha_{p_3}-\gamma^{(I)})(t-t^0)} \right) 
\nonumber \\
& & \mbox{} + \alpha_{p_2}^{(A)} \env_{A}^0 e^{(\alpha_{p_3}-\gamma^{(A)})(t-t^0)} + 
\alpha_{p_2}^{(A)} \frac{p_0^{(A)} }{\gamma^{(A)} } \left( e^{\alpha_{p_3}(t-t^0)} - e^{(\alpha_{p_3} - \gamma) (t-t^0)} \right)\nonumber \\
& & \mbox{} + \alpha_{p_2}^{(A)} \sum_{i=1}^{2} \frac{p_i^{(A)} }{\lambda_i + \gamma^{(A)} } \left(   e^{(\alpha_{p_3}+\lambda_i)(t-t^0)}  - e^{(\alpha_{p_3}-\gamma^{(A)})(t-t^0)} \right)
\eea
Integrating, using the condition $P_{r55}(t^0) = P_{r55}^0$ gives
\bea
P_{r55} & = & P_{r55}^0 e^{-\alpha_{p_3} (t-t^0)} 
 + \frac{\alpha_{p_2}^{(I)} \env_I^0 }{\alpha_{p_3} - \gamma^{(I)}} \left[  e^{-\gamma^{(I)} (t-t^0)} - e^{-\alpha_{p_3} (t-t^0)} \right]  \nonumber \\
 & & \mbox{} + \frac{\alpha_{p_2}^{(I)} p_0^{(I)} }{\gamma^{(I)}} \left[ \frac{1}{\alpha_{p_3}} \left( 1 - e^{-\alpha_{p_3} (t-t^0)} \right) - \frac{1}{\alpha_{p_3} - \gamma^{(I)} } 
 \left( e^{-\gamma^{(I)}  (t-t^0)} - e^{-\alpha_{p_3} (t-t^0)} \right) \right] \nonumber \\
 & & \mbox{} + \alpha_{p_2}^{(I)} \sum_{i=1}^{2} \frac{p_i^{(I)}}{\lambda_i + \gamma^{(I)} } \left[ \frac{ e^{\lambda_i (t-t^0)} - e^{-\alpha_{p_3} (t-t^0)} }{\alpha_{p_3} + \lambda_i} 
 - \frac{ e^{-\gamma^{(I)} (t-t^0)} - e^{-\alpha_{p_3} (t-t^0)} }{\alpha_{p_3} - \gamma^{(I)} }  \right] \nonumber \\
  & & \mbox{} + \frac{\alpha_{p_2}^{(A)} \env_A^0 }{\alpha_{p_3} - \gamma^{(A)}} \left[  e^{-\gamma^{(A)} (t-t^0)} - e^{-\alpha_{p_3} (t-t^0)} \right]  \nonumber \\
 & & \mbox{} + \frac{\alpha_{p_2}^{(A)} p_0^{(A)} }{\gamma^{(A)}} \left[ \frac{1}{\alpha_{p_3}} \left( 1 - e^{-\alpha_{p_3} (t-t^0)} \right) - \frac{1}{\alpha_{p_3} - \gamma^{(A)} } 
 \left( e^{-\gamma^{(A)}  (t-t^0)} - e^{-\alpha_{p_3} (t-t^0)} \right) \right] \nonumber \\
%
%
%
 & & \mbox{} + \alpha_{p_2}^{(A)} \sum_{i=1}^{2} \frac{p_i^{(A)}}{\lambda_i + \gamma^{(A)} } \left[ \frac{ e^{\lambda_i (t-t^0)} - e^{-\alpha_{p_3} (t-t^0)} }{\alpha_{p_3} + \lambda_i} 
 - \frac{ e^{-\gamma^{(A)} (t-t^0)} - e^{-\alpha_{p_3} (t-t^0)} }{\alpha_{p_3} - \gamma^{(A)} }  \right].
\eea
This expression has common exponential forms.  Combining these gives
\bea
P_{r55} & = & r_0 + r_1 e^{-\alpha_{p_3} (t-t^0)} + r_2^{(I)} e^{-\gamma^{(I)} (t-t^0)} + r_2^{(A)} e^{-\gamma^{(A)} (t-t^0)} + \sum_{i=1}^{2} r_{3i}  e^{\lambda_i (t-t^0)},\nonumber \\
\eea
where
\bea
r_0 & = & \frac{\alpha_{p_2}^{(I)} p_0^{(I)} }{\alpha_{p_3} \gamma^{(I)}} + \frac{\alpha_{p_2}^{(A)} p_0^{(A)} }{\alpha_{p_3} \gamma^{(A)}}, \\
r_1 & = & P_{r55}^0 - \frac{\alpha_{p_2}^{(I)} \env_I^0 }{\alpha_{p_3} - \gamma^{(I)}} - \frac{\alpha_{p_2}^{(I)} p_0^{(I)} }{\alpha_{p_3} \gamma^{(I)}} 
+\frac{\alpha_{p_2}^{(I)} p_0^{(I)}}{\gamma^{(I)}} \frac{1}{\alpha_{p_3} - \gamma^{(I)} }
\nonumber \\
 & & \mbox{} +
\alpha_{p_2}^{(I)} \sum_{i=1}^{2} \frac{p_i^{(I)}}{ \lambda_i + \gamma^{(I)} } \left[ \frac{1}{\alpha_{p_3} - \gamma^{(I)} } - \frac{1}{\alpha_{p_3} +\lambda_i}\right] \nonumber \\
 & & \mbox{} - \frac{\alpha_{p_2}^{(A)} \env_A^0 }{\alpha_{p_3} - \gamma^{(A)}} - \frac{\alpha_{p_2}^{(A)} p_0^{(A)} }{\alpha_{p_3} \gamma^{(A)}} 
 +\frac{\alpha_{p_2}^{(A)} p_0^{(A)}}{\gamma^{(A)}} \frac{1}{\alpha_{p_3} - \gamma^{(A)} }
 \nonumber \\
 & & \mbox{} +
\alpha_{p_2}^{(A)} \sum_{i=1}^{2} \frac{p_i^{(A)}}{ \lambda_i + \gamma^{(A)} } \left[ \frac{1}{\alpha_{p_3} - \gamma^{(A)} } - \frac{1}{\alpha_{p_3} +\lambda_i}\right]  \\
r_2^{(I)} & = & \frac{\env_I^0 \alpha_{p_2}^{(I)} }{\alpha_{p_3} - \gamma^{(I)}} 
-\frac{\alpha_{p_2}^{(I)} p_0^{(I)}}{\gamma^{(I)}} \frac{1}{\alpha_{p_3} - \gamma^{(I)} }
-  \frac{ \alpha_{p_2}^{(I)} }{\alpha_{p_3} - \gamma^{(I)} } \sum_{i=1}^{2} \frac{p_i^{(I)}}{ \lambda_i + \gamma^{(I)} }  \\
r_2^{(A)} & = & \frac{\env_A^0 \alpha_{p_2}^{(A)} }{\alpha_{p_3} - \gamma^{(A)}} 
-\frac{\alpha_{p_2}^{(A)} p_0^{(A)}}{\gamma^{(A)}} \frac{1}{\alpha_{p_3} - \gamma^{(A)} }
- \frac{ \alpha_{p_2}^{(A)} }{\alpha_{p_3} - \gamma^{(A)} } \sum_{i=1}^{2} \frac{p_i^{(A)}}{ \lambda_i + \gamma^{(A)} } \\
r_{3i} & = & \alpha_{p_2}^{(I)} \frac{p_i^{(I)}}{ (\lambda_i + \gamma^{(I)} )(\alpha_{p_3} + \lambda_i)} +
\alpha_{p_2}^{(A)} \frac{p_i^{(A)}}{ (\lambda_i + \gamma^{(A)} )(\alpha_{p_3} + \lambda_i)}
\eea

Other cases (e.g.~$\gamma^{I}=0$, $\gamma^{A}=0$, $\alpha_{p_3}=\gamma^{I}$, $\alpha_{p_3}=\gamma^{A}$, etc. can be addressed in a similar manner as necessary.

\textbf{$P_{24}$ equation}

\bea
\label{eq:P24ode}
\frac{dP_{24}}{dt} + \gamma_{p_2} P_{24}& = & \alpha_{p_3} P_{r55}   
\eea
Using the result for $P_{r55}$ leads to the equation
\bea
\frac{dP_{24}}{dt} + \gamma_{p_2} P_{24}& = & \alpha_{p_3} r_0 + \alpha_{p_3} r_1 e^{-\alpha_{p_3} (t-t^0)} \\
& & \mbox{} + \alpha_{p_3} r_2^{(I)} e^{-\gamma^{(I)} (t-t^0)} + \alpha_{p_3} r_2^{(A)} e^{-\gamma^{(A)} (t-t^0)} + \alpha_{p_3} \sum_{i=1}^{2} r_{3i}  e^{\lambda_i (t-t^0)}.
\nonumber
\eea

\textbf{cases} $\gamma_{p_2} \neq 0$, $\gamma^{(I)} \neq 0$, $\gamma^{(A)} \neq 0$ and $\gamma_{p_2} \neq \gamma^{(I)}$, $\gamma_{p_2} \neq \gamma^{(A)}$

Multiplying equation~(\ref{eq:P24ode}) through by the integrating factor $e^{\gamma_{p_2} t}$ leads to
\bea
\frac{d}{dt} \left( e^{\gamma_{p_2} (t-t^0)} P_{24} \right) & = & \alpha_{p_3} r_0 e^{\gamma_{p_2} (t-t^0)} 
+ \alpha_{p_3} r_1 e^{(\gamma_{p_2} - \alpha_{p_3}) (t-t^0)} \nonumber \\
& & \mbox{} + \alpha_{p_3} r_2^{(I)} e^{(\gamma_{p_2} -\gamma^{(I)}) (t-t^0)}
+ \alpha_{p_3} r_2^{(A)} e^{(\gamma_{p_2} -\gamma^{(A)}) (t-t^0)} \nonumber \\
& & \mbox{}  + \alpha_{p_3} \sum_{i=1}^{2} r_{3i}  e^{(\gamma_{p_2} +\lambda_i )(t-t^0)},
\eea
Upon integrating and using $P_{24}(t^0)=P_{24}^0$ we find 
\bea
P_{24} & = & P_{24}^0 e^{-\gamma_{p_2} (t-t^0)} + \frac{\alpha_{p_3} r_0}{\gamma_{p_2}} \left[ 1- e^{-\gamma_{p_2} (t-t^0)} \right] \nonumber \\
 & & \mbox{} + \frac{\alpha_{p_3} r_1}{\gamma_{p_2} - \alpha_{p_3} } \left[ e^{-\alpha_{p_3}(t-t^0)} - e^{-\gamma_{p_2}(t-t^0)} \right] \nonumber \\
 & & \mbox{} + \frac{\alpha_{p_3} r_2^{(I)} }{\gamma_{p_2} - \gamma^{(I)} } \left[ e^{-\gamma^{(I)}(t-t^0)} - e^{-\gamma_{p_2}(t-t^0)}  \right] \nonumber \\
 & & \mbox{} + \frac{\alpha_{p_3} r_2^{(A)} }{\gamma_{p_2} - \gamma^{(A)} } \left[ e^{-\gamma^{(A)}(t-t^0)} - e^{-\gamma_{p_2}(t-t^0)}  \right] \nonumber \\
 & & \mbox{} + \alpha_{p_3} \sum_{i=1}^{2} \frac{r_{3i}}{\gamma_{p_2} + \lambda_i}
 \left[ e^{\lambda_i (t-t^0)} - e^{-\gamma_{p_2}(t-t^0)}  \right]
\eea
Combining terms gives the solution
\bea
P_{24} & = & s_0 + s_1 e^{-\gamma_{p_2} (t-t^0)} + s_2 e^{-\alpha_{p_3} (t-t^0)} \nonumber \\
& & \mbox{} + s_3^{(I)} e^{-\gamma^{(I)} (t-t^0)} + 
s_3^{(A)} e^{-\gamma^{(A)} (t-t^0)} + \sum_{i=1}^{2} s_{4i} e^{\lambda_i (t-t^0)},
\eea
where
\bea
s_0 & = & \frac{\alpha_{p_3} r_0}{\gamma_{p_2}} \\
s_1 & = & P_{24}^0 - \frac{\alpha_{p_3} r_0}{\gamma_{p_2}} - \frac{\alpha_{p_3} r_1}{ \gamma_{p_2} - \alpha_{p_3}} - \frac{\alpha_{p_3} r_2^{(I)} }{\gamma_{p_2} - \gamma^{(I)} }
- \frac{\alpha_{p_3} r_2^{(A)} }{\gamma_{p_2} - \gamma^{(A)} } - \alpha_{p_3} \sum_{i=1}^{2} \frac{r_{3i}}{\gamma_{p_2} + \lambda_i} \nonumber \\ \\
s_2 & = & \frac{\alpha_{p_3} r_1}{\gamma_{p_2} - \alpha_{p_3}} \\
s_3^{(I)} & = & \frac{\alpha_{p_3} r_2^{(I)}}{\gamma_{p_2} - \gamma^{(I)}} \\
s_3^{(A)} & = & \frac{\alpha_{p_3} r_2^{(A)}}{\gamma_{p_2} - \gamma^{(A)}} \\
s_{4i} & = & \alpha_{p_3} \frac{r_{3i}}{\gamma_{p_2} + \lambda_i}.
\eea

\textbf{case} $\gamma_{p_2} = 0$, $\gamma^{(I)} \neq 0$, $\gamma^{(A)} \neq 0$

Here
\bea
\frac{d P_{24}}{dt} & = & \alpha_{p_3} r_0  
+ \alpha_{p_3} r_1 e^{- \alpha_{p_3} (t-t^0)} \nonumber \\
& & \mbox{} + \alpha_{p_3} r_2^{(I)} e^{ -\gamma^{(I)} (t-t^0)}
+ \alpha_{p_3} r_2^{(A)} e^{ -\gamma^{(A)} (t-t^0)} \nonumber \\
& & \mbox{}  + \alpha_{p_3} \sum_{i=1}^{2} r_{3i}  e^{\lambda_i (t-t^0)}.
\eea
Integrating and applying $P_{24}(t^0)=P_{24}^0$ gives
\bea
P_{24} & = & P_{24}^0 + \alpha_{p_3} r_0 (t-t^0) - r_1 \left[ e^{-\alpha_{p_3}(t-t^0)}  - 1 \right]  
- \frac{\alpha_{p_3} r_2^{(I)} }{ \gamma^{(I)} } \left[ e^{-\gamma^{(I)} (t-t^0)}  - 1 \right] \nonumber \\
 & & \mbox{}  
 - \frac{\alpha_{p_3} r_2^{(A)} }{ \gamma^{(A)} } \left[ e^{-\gamma^{(A)} (t-t^0)}  - 1 \right]
+ \alpha_{p_3} \sum_{i=1}^{2} \frac{r_{3i}}{\lambda_i} \left[ e^{\lambda_i (t-t^0)} - 1\right]
\eea

We can write this as 
\bea
P_{24} & = & s_0 + s_1 (t-t^0) + s_2 e^{-\alpha_{p_3} (t-t^0)} \nonumber \\
& & \mbox{} + s_3^{(I)} e^{-\gamma^{(I)} (t-t^0)} + 
s_3^{(A)} e^{-\gamma^{(A)} (t-t^0)} + \sum_{i=1}^{2} s_{4i} e^{\lambda_i (t-t^0)},
\eea
where
\bea
s_0 & = & P_{24}^0 + r_1 + \frac{\alpha_{p_3} r_2^{(I)} }{ \gamma^{(I)} } + \frac{\alpha_{p_3} r_2^{(A)} }{ \gamma^{(A)} } - \alpha_{p_3} \sum_{i=1}^{2}  \frac{r_{3i}}{\lambda_i} \\
s_1 & = & \alpha_{p_3} r_0 ,\\
s_2 & = & -r_1,\\
s_3^{(I)} & = & - \frac{\alpha_{p_3} r_2^{(I)} }{ \gamma^{(I)} },\\
s_3^{(A)} & = & - \frac{\alpha_{p_3} r_2^{(A)} }{ \gamma^{(A)} },\\
s_{4i} & = & \alpha_{p_3} \frac{r_{3i}}{\lambda_i}.
\eea

\bibliographystyle{plainnat}
\bibliography{ref.bib}   


\end{document}